\documentclass[10pt]{article}
\usepackage[svgnames]{xcolor}
\usepackage[breaklinks,hidelinks]{hyperref}
\usepackage{knowledge}
\usepackage{pgf} 
\usepackage{tikz} 
\usetikzlibrary{arrows,automata}
\usepackage{cite} 
\usepackage{amsmath} 
\usepackage{amssymb}
\usepackage{amsthm} 
\usepackage{amsfonts}
\usepackage{array} 
\usepackage{stfloats}
\usepackage{tikz-cd} 
\usetikzlibrary{matrix,arrows}
\usepackage{etoolbox}
\knowledgedirective{notion}{autoref,style=notion,intro style=intro notion}
\knowledgestyle{notion}{color=black} 
\knowledgestyle{intro notion}{emphasize,color=black} 
\knowledgedirective{mathsymb}{autoref,style=mathsymb, intro style=intro mathsymb}
\knowledgestyle{mathsymb}{color=black} 
\knowledgestyle{intro mathsymb}{color=black} 
\knowledgestyle{unknown}{color=black} 
\knowledgestyle{unknown}{color=black} 

\knowledge{marked letters}[marked letter]{notion}
\knowledge{modular quantifier}{notion}
\knowledge {modular quantifiers}{synonym}
\knowledge{Boolean space}[Boolean spaces|Boolean (Stone) spaces]{notion}
\knowledge{Stone-\v{C}ech compactification}{notion}
\knowledge{Vietoris functor}[Vietoris (hyper)space|Vietoris
  hyperspace|Vietoris monad|Vietoris|Vietoris space|Vietoris
  spaces]{notion}
\knowledge{Boolean space with an internal monoid}[\swim|\swims|Boolean
  spaces with internal monoids]{notion}
\knowledge{Eilenberg-Moore algebra}[Eilenberg-Moore
  algebras|$T$-algebra|$T$-algebras]{notion}
\knowledge{semiring}[semirings]{notion}
\knowledge{free $S$-semimodule monad}[{free semimodule monad}|{free $S$-semimodule monads}]{notion}
\knowledge{finitary}{notion}
\knowledge{measure}[$S$-valued measure|valued measure|measures]{notion}
\knowledge{commutative}[commutative monad|commutativity]{notion}
\knowledge{profinite monad}{link=\wT,color=DarkBlue}
\knowledge{biaction}[biactions]{notion}
\knowledge{codensity monad}[codensity monads]{notion}
\knowledge{morphism}[morphisms]{notion} 
\knowledge{recognises}[recognised|recognising|recogniser|recognisers]{notion}
\knowledge{Kleisli map}[Kleisli maps]{notion} 

\knowledge{density}[densities]{notion}

\knowledge{length preserving}{notion}

\knowledge{\Set}{color=black,autoref} 
\knowledge{\Setfin}{color=black,autoref} 
\knowledge{\Bool}{color=black,autoref} 
\knowledge{\Stone}{color=black,autoref} 
\knowledge\AlgCT{color=black,autoref} 
\knowledge\AlgfinT{color=black,autoref} 
\knowledge\swimcat{color=black,autoref} 
\knowledge{\Kl}{mathsymb} 

\knowledge{\P}{mathsymb} 
\knowledge{\Pfin}{mathsymb} 
\knowledge{\I}{mathsymb} 
\knowledge{\J}{mathsymb} 
\knowledge{\wS}{mathsymb} 
\knowledge{\wT}{mathsymb} 
\knowledge{\fgt}{mathsymb} 
\knowledge{\wfgt}{mathsymb} 
\knowledge{\beta}{mathsymb} 
\knowledge{\wbeta}{mathsymb} 

\knowledge{\t}{mathsymb} 
\knowledge{\thash}{mathsymb} 
\knowledge{\wh}{mathsymb}   
\knowledge{\wB}{mathsymb}

\knowledge{\Qk}[\Qkp]{mathsymb} 
\knowledge{\Cl}{color=black,autoref} 
\knowledge{\B}{mathsymb} 
\knowledge{\Breg}{mathsymb} 
\knowledge{\Bzero}{mathsymb} 
\knowledge{\Bprime}{mathsymb} 
\knowledge{\QB}{mathsymb} 
\knowledge{\SB}{mathsymb} 
\knowledge{\SBp}{mathsymb} 
\knowledge{\DsM}{mathsymb} 
\knowledge{\DsX}{mathsymb} 
\knowledge{\Dphi}{mathsymb} 
\knowledge{\DssM}{mathsymb} 
\knowledge{\DssX}{mathsymb} 
\knowledge{\Dssphi}{mathsymb} 
\knowledge{\mw}{mathsymb} 
\knowledge{\wz}[\uz]{mathsymb} 
\knowledge{\len}{mathsymb} 
\knowledge{\embz}{mathsymb} 
\knowledge{\Dwphi}{mathsymb} 
\knowledge{\inte}{mathsymb} 
\knowledge{\Lex}{mathsymb} 
\knowledge{\cf}{mathsymb} 
\knowledge{\phiQ}{mathsymb} 
\knowledge{\wphiQ}{mathsymb} 
\knowledge{\phiz}[\wphiz]{mathsymb} 
\knowledge{\mat}{mathsymb} 
\knowledge{\Rmon}{mathsymb}
\knowledge{\wRmon}{mathsymb}
\knowledge{\X}{mathsymb} 
\knowledge{\synMon}{mathsymb} 
\knowledge{\Ac}{mathsymb} 
\knowledge{\Dv}{mathsymb} 
\knowledge{\Lz}{mathsymb} 
\knowledge{\fw}[\fu]{mathsymb} 

\knowledge{\As}[\Ats|\Bs]{mathsymb} 

\knowledge\Zq{color=black,autoref}  
\knowledge\Lmodex{color=black,autoref} 

\knowledge\semmon{color=black,autoref} 
\knowledge{\Sm}{color=black,autoref} 
\knowledge\V{color=black,autoref} 

\usepackage{fullpage}

\usepackage[affil-it]{authblk}

\theoremstyle{plain}
\newtheorem{theorem}{Theorem}[section]
\newtheorem{lemma}[theorem]{Lemma}
\newtheorem{corollary}[theorem]{Corollary}
\newtheorem{proposition}[theorem]{Proposition}
\theoremstyle{definition}
\newtheorem{definition}[theorem]{Definition}
\newtheorem{remark}[theorem]{Remark}
\newtheorem{example}[theorem]{Example}
\newtheorem{notation}[theorem]{Notation}

\newrobustcmd{\C}{\mathcal{C}} 
\newrobustcmd{\Set}{\kl[\Set]{\mathsf{Set}}} 
\newrobustcmd{\Setfin}{\kl[\Setfin]{{\mathsf{Set}_{f}}}} 
\newrobustcmd{\Bool}{\kl[\Bool]{\mathsf{Boole}}} 
\newrobustcmd{\Stone}{\kl[\Stone]{\mathsf{BStone}}} 
\newrobustcmd{\Alg}[2]{{#1}^{#2}} 
\newrobustcmd{\AlgCT}{\kl[\AlgCT]{{\C}^{T}}}
\newrobustcmd{\AlgT}{\kl[\AlgCT]{{\mathsf{Set}^{T}}}}
\newrobustcmd{\Algfin}[2]{{#1}^{#2}_{f}} 
\newrobustcmd{\AlgfinT}{\kl[\AlgfinT]{{\Algfin{\mathsf{Set}}{T}}}} 
\newrobustcmd{\swimcat}{\kl[\swimcat]{{\sf BiM}}}

\renewrobustcmd{\P}{\kl[\P]{\mathcal{P}}} 
\newrobustcmd{\Pfin}{\kl[\Pfin]{\mathcal{P}_{f}}} 
\newrobustcmd{\V}{\kl[\V]{\mathcal{V}}} 
\newrobustcmd{\Sm}{\kl[\Sm]{\mathcal{S}}} 
\newrobustcmd{\I}{\kl[\I]I} 
\newrobustcmd{\J}{\kl[\J]J} 
\DeclareMathOperator{\Ran}{\mathrm Ran} 
\newrobustcmd{\w}[1]{\widehat{#1}} 
\newrobustcmd{\wh}{\kl[\wh]{\w{h}}}
\newrobustcmd{\wB}{\kl[\wB]{\w{B}}} 
\newrobustcmd\wS {\kl[\wS]{\w{\mathcal S}}} 
\newrobustcmd{\wT}{\kl[\wT]{\w{T}}} 
\newrobustcmd{\fgt}{\kl[\fgt]{|-|}} 
\newrobustcmd{\wfgt}{\kl[\wfgt]{\w{|-|}}} 
\newrobustcmd{\wbeta}{\kl[\wbeta]{\w{\beta}}} 

\newrobustcmd{\mono}{\rightarrowtail} 
\newrobustcmd{\epi}{\twoheadrightarrow} 
\renewrobustcmd{\t}{\kl[\t]\tau}  
\newrobustcmd{\thash}{\kl[\thash]{\tau^\#}}  
\newrobustcmd{\embz}{\kl[\embz]{(\ )^0}} 

\newrobustcmd{\Dv}{\kl[\Dv]{\scalebox{1.2}{$\diamond$}}} 
\newrobustcmd{\Ds}{\Diamond}
\newrobustcmd{\bDs}{\scalebox{1.2}{$\Diamond$}}
\newrobustcmd{\boxa}{\scalebox{0.7}{\raisebox{1pt}{$\Box$}}}

\newrobustcmd{\X}{\kl[\X]{X}} 

\newrobustcmd{\As}{\kl[\As]{{A^*}}} 
\newrobustcmd{\Bs}{\kl[\As]{{B^*}}} 
\newrobustcmd{\Ats}{\kl[\As]{{(A\times 2)^*}}} 

\newrobustcmd{\Qk}{\kl[\Qk]{{\mathcal Q}_k}} 
\newrobustcmd{\Qkp}{\kl[\Qkp]{{\mathcal Q}_{k_1}}} 
\newrobustcmd{\Cl}{\kl[\Cl]{\it Clop}} 
\newrobustcmd{\B}{\kl[\B]{\mathcal{B}}} 
\newrobustcmd{\Bzero}{\kl[\Bzero]{\mathcal{B}_0}} 
\newrobustcmd{\Breg}{\kl[\Breg]{\mathcal{B}}} 
\newrobustcmd{\Bprime}{\kl[\Bprime]{\mathcal{B}'}} 
\newrobustcmd{\QB}{\kl[\QB]{\mathcal{QB}}} 
\newrobustcmd{\DssM}{\kl[\DssM]{\Ds_S M}} 
\newrobustcmd{\DssX}{\kl[\DssM]{\Ds_S X}} 
\newrobustcmd{\Dssphi}{\kl[\Dssphi]{\Ds_S \phi}} 
\newrobustcmd{\DsM}{\kl[\DsM]{\Ds M}} 
\newrobustcmd{\DsX}{\kl[\DsM]{\Ds X}} 
\newrobustcmd{\Dphi}{\kl[\Dphi]{\Ds \phi}} 
\newrobustcmd{\Dwphi}{\kl[\Dwphi]{\Ds \tilde{\phi}}} 
\newrobustcmd{\phiQ}{\kl[\phiQ]{\phi_Q}} 
\newrobustcmd{\wphiQ}{\kl[\wphiQ]{\tilde{\phi}_Q}} 
\newrobustcmd{\phiz}{\kl[\phiz]{\phi_0}} 
\newrobustcmd{\wphiz}{\kl[\wphiz]{\tilde{\phi}_0}} 
\newrobustcmd{\mw}[1]{\kl[\mw]{w^{({#1})}}} 
\newrobustcmd{\wz}{\kl[\wz]{w^{0}}} 
\newrobustcmd{\uz}{\kl[\wz]{u^{0}}} 
\newrobustcmd{\len}[1]{\kl[\len]{|{#1}|}} 
\newrobustcmd \inte {\kl[\inte]{\int}} 
\newrobustcmd{\Lz}{\kl[\Lz]{L_0}} 
\newrobustcmd{\fw}{\kl[\fw]{f_w}} 
\newrobustcmd{\fu}{\kl[\fw]{f_u}} 

\newcommand{\wphi}{\tilde{\phi}} 

\newrobustcmd{\N}{\mathbb{N}} 
\newrobustcmd{\Z}{\mathbb{Z}} 
\newrobustcmd{\2}{{2}} 
\newrobustcmd{\id}{\mathsf{id}} 
\newrobustcmd{\op}{\mathsf{op}} 

\DeclareMathOperator{\colim}{\mathrm{colim}} 

\newrobustcmd{\BA}{\sc{BA}} 
\newrobustcmd{\BAs}{{\sc BA}s} 
\newrobustcmd{\BAq}{{\sc BA}q} 
\newrobustcmd{\swim}{\kl[Boolean space with an internal monoid]{{\sc
      B}{\rm i}{\sc M}}}
\newrobustcmd{\swims}{\kl[Boolean space with an internal monoid]{{\sc
    B}{\rm i}{\sc M}s}}

\renewrobustcmd{\bar}{\overline}

\newrobustcmd{\st}{\mathsf{st}} 
\newrobustcmd{\ev}{\mathsf{ev}} 
\newrobustcmd{\Rmon}{\kl[\Rmon]{\mathsf{R}}}
\newrobustcmd{\wRmon}{\kl[\wRmon]{\w{\mathsf{R}}}}
\newrobustcmd{\mat}[2]{\kl[\mat]{\mathcal{M}_{#1}{(#2)}}} 
\newrobustcmd{\SB}[2]{\kl[\SB]{[{#1},{#2}]}} 
\newrobustcmd{\SBp}[2]{\kl[\SBp]{\overline{[{#1},{#2}]}}} 
\newrobustcmd{\cf}[2]{\kl[\cf]{[{#1},{#2}]}} 
\newrobustcmd{\Ac}[2]{\kl[\Ac]{{#1}_{{#2},{#2}}}} 
\newrobustcmd{\TR}{\mathcal{T}_R} \newrobustcmd{\modexists}{{\exists_{p
      \textrm{ mod }q}}} \newrobustcmd{\Kl}{\kl[\Kl]{\mathsf{Kl}}}
\newrobustcmd{\synMon}{\kl[\synMon]{\As /{\sim_L}}} 

\renewrobustcmd{\Im}{\mathsf{Im}}
\newrobustcmd\Lex{\kl[\Lex]{L_\exists}}
\newrobustcmd\Lmodex{\kl[\Lmodex]{L_\modexists}}
\newrobustcmd\Zq{\kl[\Zq]{\mathbb Z_q}}

\renewcommand{\epsilon}{\varepsilon}

\begin{document}
\title{Quantifiers on languages and codensity monads\thanks{This project has received funding from the European
    Research Council (ERC) under the European Union's Horizon 2020
    research and innovation programme (grant agreement No.670624).
    The third author also acknowledges financial support from Sorbonne Paris Cit{\'e} (PhD agreement USPC IDEX -- REGGI15RDXMTSPC1GEHRKE).}}
\author[1]{Mai Gehrke}
\author[2]{Daniela Petri{\c s}an}
\author[3]{Luca Reggio}
\affil[1]{\small Laboratoire J.\ A.\ Dieudonn{\'e}, CNRS and Universit{\'e} C{\^o}te d'Azur, France, \textnormal{email: \texttt{mai.gehrke@unice.fr}}}
\affil[2]{\small IRIF, CNRS and Universit{\'e} Paris Diderot, France, \textnormal{email: \texttt{petrisan@irif.fr}}}
\affil[3]{\small Laboratoire J.\ A.\ Dieudonn{\'e}, Universit{\'e} C{\^o}te d'Azur, France and \newline Institute of Computer Science, Czech Academy of Sciences, Czech Republic, \textnormal{email: \texttt{reggio@cs.cas.cz}}}
\date{}

\maketitle

\begin{abstract}
  This paper contributes to the techniques of topo-algebraic
  recognition for languages beyond the regular setting as they relate
  to logic on words.  In particular, we provide a general construction
  on recognisers corresponding to adding one layer of various kinds of
  quantifiers and prove a corresponding Reutenauer-type theorem.  Our main
  tools are \kl{codensity monads} and duality theory.  Our
  construction hinges on a measure-theoretic characterisation of the
  profinite monad of the \kl{free $S$-semimodule monad} for finite and
  commutative \kl{semirings} $S$, which generalises our earlier
  insight that the \kl{Vietoris monad} on \kl{Boolean spaces} is the
  \kl{codensity monad} of the finite powerset functor.
\end{abstract}

\section{Introduction}
It is well known that the combinatorial property of a language of
being given by a star-free regular expression can be described both by
algebraic and by logical means. Indeed, on the algebraic side, the
star-free languages are exactly those languages whose syntactic
monoids do not contain any non-trivial groups as subsemigroups. On the logical side,
properties of words can be expressed in predicate logic by considering
variables as positions in the word, relation symbols asserting that a
position in a word has a certain letter of the alphabet, and possibly
additional predicates on positions (known as numerical predicates).
As shown by McNaughton and Papert in \cite{MP1971}, the class of
languages definable by first-order sentences over the numerical
predicate $<$ consists precisely of the star-free ones.

The theory of formal languages abounds with such results showing the
strong interplay between logic and algebra. For instance, Straubing,
Th\'erien and Thomas introduced in \cite{STRAUBINGTT} a class of additional
quantifiers, the so-called modular quantifiers $\modexists$. (Recall
that a word satisfies a formula $\modexists x.\varphi(x)$ provided the
number of positions $x$ for which $\varphi(x)$ holds is congruent to
$p$ modulo $q$). There it is shown for example that the languages
definable using modular quantifiers of modulus $q$ are exactly the
languages whose syntactic monoids are solvable groups of cardinality
dividing a power of $q$.

Studying modular quantifiers is relevant for tackling open problems in
Boolean circuit complexity, see for example~\cite{StraubingT08} for a
discussion.  Since Boolean circuit classes contain non-regular
languages, expanding the automata theoretic techniques beyond the
regular setting is also relevant for addressing these problems.

A fundamental tool in studying the connection between algebra and
logic in this setting is the availability of constructions on monoids
which mirror the action of quantifiers. That is, given the syntactic
monoid for a language with a free variable one wants to construct a
monoid which recognises the quantified language. Constructions of this
type abound, and are all versions of semidirect products, with the
block product playing a central r\^ole as it allows one to construct
recognisers for many different quantifiers \cite{TeTh07}.

The present article is an expanded and improved version of the publication \cite{GPR2017}, where the main results were first announced. 
Its purpose is to expand the techniques available for monoids and provide
the topo-algebraic characterisation of adding one layer of various kinds of
quantifiers, beyond the regular setting.  A first step was made in~\cite{GehrkePR16}, where a) we introduced a topological
notion of recogniser, that will be motivated in the next subsection,
and b) we gave a notion of unary Sch\"utzenberger product that
corresponds, on the recogniser side, to adding one layer of the existential quantifier for
arbitrary languages of words.

In Section~\ref{subsec:intro-dual} we provide a gentle introduction
and motivate the duality-theoretic approach to language
recognition. In Section~\ref{subsec:intro-cod} we present
\kl{codensity monads}, our tool of choice for systematically obtaining
the relevant topological constructions, and we briefly discuss related
work. Finally, in Section~\ref{subsec:intro-cont} we present the main
contributions of this paper and provide an overview of the remainder of the paper.

\subsection{Duality for language recognition}\label{subsec:intro-dual}
Stone duality plays an important r\^ole and has a long tradition in
many areas of semantics, e.g.\ in domain theory and in modal
logic. In~\cite{Pippenger1997} Pippenger made explicit the link between
Stone duality and regular languages, by proving that the Boolean
algebra of regular languages over a finite alphabet $A$ is the dual of
the free profinite monoid on $A$. Yet, only recently, starting with
the papers~\cite{GehrkeGP10, GehrkeGP08}, the deep connection between
this field and formal language theory started to emerge. In these
papers a new notion of language recognition, based on topological
methods, was proposed for the setting of non-regular
languages. Moreover, the scene was set for a new duality-theoretic
understanding of the celebrated Eilenberg-Reiterman theorems,
establishing a connection between varieties of languages,
pseudo-varieties of finite algebras and profinite equations. This led
to an active research area where categorical and duality-theoretic
methods are used to encompass notions of language recognition for
various automata models. See for example the monadic approach to
language recognition put forward by Boja{\'n}czyk~\cite{Bojanczyk15},
or the series of papers on a category-theoretic approach to
Eilenberg-Reiterman theory (\hspace{1sp}\cite{AdamekMUM15} and
references herein).

\AP Let us illustrate the interplay between duality theory and the
theory of regular languages by explaining the duality between the
syntactic monoid of a regular language $L$ on a finite alphabet $A$,
and the Boolean subalgebra $\intro*\Breg\hookrightarrow \P(\As)$
generated by the quotients of $L$, i.e., by the sets
\[
w^{-1}Lv^{-1}=\{u\in \As\mid wuv \in L\}
\]
for $w,v\in A^*$. In this setting one makes use only of the finite
duality between the category of finite Boolean algebras and the
category of finite sets which, at the level of objects, asserts that
each finite Boolean algebra is isomorphic to the powerset of its
atoms.

Since the language $L$ is regular it has only finitely many
quotients, say
\[
\{w_1^{-1}Lv_1^{-1},\ldots, w_n^{-1}Lv_n^{-1}\}.
\]
The finite Boolean algebra generated by this set has as atoms the
non-empty subsets of $\As$ of the form
\[
\bigcap_{i\in I}w_i^{-1}Lv_i^{-1} \cap \bigcap_{j\in
  J}(w_j^{-1}Lv_j^{-1})^c
\]
for some partition $I\cup J$ of $\{1,\ldots,n\}$. \AP We clearly see
that such atoms are in one-to-one correspondence with the equivalence
classes of the Myhill syntactic congruence $\sim_L$, and thus with the
elements of the syntactic monoid $\intro{\synMon}$ of $L$.

\AP However, the more interesting aspect of this approach is that one
can also explain the monoid structure of $\synMon$ and the syntactic
morphism in duality-theoretic terms. For this, we have to recall first
the duality between the category of sets and the category of complete
atomic Boolean algebras. At the level of objects, every complete
atomic Boolean algebra is isomorphic to the \AP powerset of its
atoms. So the dual of $\As$ is $\intro*\P(\As)$, but the duality also
tells us that quotients on one side are turned into embeddings on the
other. Thus, we have the duality between the following morphisms
\[
\begin{tikzcd}
  \Breg\ar[r,hookrightarrow] &\P(\As) & {|} & \As \ar[r,two heads] &
  \synMon.
\end{tikzcd}
\]

Also, the left action of $\As$ on itself given by appending a word $w$
on the left corresponds, on the dual side, to a left quotient
operation (which is a right action):
\[
\begin{tikzcd}[row sep=0pt]
  \P(\As)\ar[r,"\Lambda_w"] &\P(\As) &{}\ar[d,dash]  & \As \ar[r,"l_w"] & \As \\
  U\ar[r,mapsto] & w^{-1}U &{} & v\ar[r,mapsto] & wv
\end{tikzcd}
\]
Since the Boolean algebra $\Breg$ is closed under quotients, and thus
we have commuting squares as the left one in diagram~\eqref{eq:14}, by
duality we obtain a left action of $\As$ on $\synMon$. By an analogous
argument, one also obtains a right action of $\As$ on $\synMon$ and
the two actions commute. It is a simple lemma, see~\cite{GehrkePR16},
that since $\synMon$ is a quotient of $\As$ and it is equipped with
commuting left and right $\As$-actions (called in loc.\ cit.\ an
\emph{$\As$-\kl{biaction}}), then one can uniquely define a monoid
multiplication on $\synMon$ so that the quotient
$\As\twoheadrightarrow \synMon$ is a monoid morphism.
\begin{equation}
  \label{eq:14}
  \begin{tikzcd}
    \Breg\ar[r,hookrightarrow]\ar[d,swap, dashed,"\Lambda_w"] &\P(\As)\ar[d,"\Lambda_w"] &{}\ar[d,dash] &  \As\ar[d,"l_w", swap] \ar[r,two heads] & \synMon\ar[d, dashed,"l_w"] \\
    \Breg\ar[r,hookrightarrow] &\P(\As) &{}& \As \ar[r,two heads] &
    \synMon
  \end{tikzcd}
\end{equation}

This approach paves the way to a notion of recogniser and syntactic object 
pertinent for non-regular languages. In the case of a non-regular language  $L$, 
the Boolean algebra $\Breg$ spanned by the quotients of $L$
is no longer finite, so the finite or discrete duality theorems we have employed
previously are no longer applicable. Instead, we use the full power of
Stone duality, which establishes the dual equivalence between the
category of Boolean algebras and the category $\Stone$ of \kl{Boolean
  (Stone) spaces}, that is, zero-dimensional compact Hausdorff 
  spaces. In this setting, the dual
of $\P(\As)$ is the \kl{Stone-\v{C}ech compactification}
$\kl{\beta}(\As)$ of the discrete space $\As$.  \AP The embedding of
$\Breg$ into $\P(\As)$ is turned by the duality theorem into a
quotient of topological spaces as displayed below, where we denote the dual of $\Breg$ by
$\intro{\X}$.
\[
\begin{tikzcd}
  \Breg\ar[r,hookrightarrow] &\P(\As) & {|} & \kl{\beta}(\As)
  \ar[r,two heads] & \X
\end{tikzcd}
\]
The syntactic monoid of the language $L$, now infinite, can be seen as
a dense subset of $\X$, and is indeed the image of the composite map
$\As\hookrightarrow \kl{\beta}(\As) \twoheadrightarrow \X$ where the
first arrow is the embedding of $\As$ in its \kl{Stone-\v{C}ech
  compactification}. We thus obtain a commuting diagram as follows.
\[
\begin{tikzcd}
  \kl{\beta}(\As)\ar[r,two heads] & \X \\
  \As\ar[u,hook]\ar[r,two heads] & \synMon\ar[u, hook]
\end{tikzcd}
\]
Furthermore, one can show that the syntactic monoid acts
(continuously) on $\X$ both on the left and on the right, and these actions
commute. This led us, in~\cite{GehrkePR16}, to the definition of a 
\emph{\kl{Boolean space with
    an internal monoid}} ({\swim}) as a suitable notion for language
recognition beyond the regular setting. We recall this (in fact a
small variation of it) in Definition~\ref{def:swim}.

\subsection{Profinite monads}\label{subsec:intro-cod}

\AP Profinite methods have a long tradition in language theory, see
for example~\cite{ALMEIDA19981}. To accommodate these tools in his
monadic approach to language recognition,
Boja\'nczyk~\cite{Bojanczyk15} has recently introduced a construction
transforming a monad $T$ on $\intro*\Set$ (the category of sets and
functions) into a so-called \emph{profinite monad}, again on the
category of sets. The latter monad allowed him to study in this generic framework
the profinite version of the objects modelled by $T$, such as
profinite words, profinite countable chains and profinite trees.

A very much related construction of a \emph{profinite monad of $T$}
was introduced in~\cite{ChenAMU16}, this time as a monad on the
category of \kl{Boolean spaces}, obtained as a so-called
\emph{\kl{codensity monad}} for a functor from the category of
finitely carried \kl{$T$-algebras} to \kl{Boolean spaces}, that we
describe in the next section.

The \kl{codensity monad} is a standard construction in category
theory, going back to the work of Kock in the 60s.  It is well known
that any right adjoint functor $G$ induces a monad obtained by
composition with its left adjoint, and this is exactly the
\kl{codensity monad} of $G$. In general, the \kl{codensity monad} of a
functor which is not necessarily right adjoint, provided it exists, is
the best approximation to this phenomenon.  \AP For example the
\kl{codensity monad} of the forgetful functor
$\intro*\fgt\colon\Stone\to\Set$ on \kl{Boolean spaces} is the
ultrafilter monad on $\Set$ obtained by composition with its left
adjoint $\kl{\beta}\colon\Set\to\Stone$. \AP The same monad has yet
another description as a \kl{codensity monad}, this time for the
inclusion of the category $\intro*\Setfin$ of finite sets into $\Set$,
a fact proved in~\cite{KENNISON1971317} and recently revisited in the
elegant paper~\cite{Leinster2010}.

\AP The starting point of the present paper is the observation that
the unary Sch\"utzenberger product $(\Ds X, \Ds M)$ of a {\swim} $(X,M)$
from our paper~\cite{GehrkePR16} hinges, at a deeper level, on the
fact that the \emph{\kl{Vietoris monad}} $\V$ on the category of
\kl{Boolean spaces} (which is heavily featured in that construction)
is the profinite monad of the finite powerset monad $\intro\Pfin$ on
$\Set$. Recall that any \kl{Boolean space} $X$ is the cofiltered (or
inverse) limit of its finite quotients $X_i$. Then one can check that
the Vietoris space $\V X$ can be obtained as the cofiltered limit of
the finite sets $\Pfin X_i$.

In order to find suitable recognisers for languages quantified by,
e.g., modular existential quantifiers, we need a slightly different
construction than $(\Ds X,\Ds M)$ of~\cite{GehrkePR16}. Specifically, we
observe that the semantics of these quantifiers can be modelled, at
least at the level of finite monoids, by the \kl{free $S$-semimodule
  monad} $\Sm$, for a suitable choice of the \kl{semiring} $S$.  It should
be noted that $\Pfin$ is also an instance of the \kl{free
  $S$-semimodule monad}, for the Boolean semiring $\2$.  To obtain
corresponding constructions at the level of \kl{Boolean spaces with
  internal monoids}, one needs to understand the analogue of the
\kl{Vietoris} construction for the monad $\Sm$. And the obvious
candidate, from a category-theoretic perspective, is the \kl{codensity
  monad} of $\Sm$.

\subsection{Contributions}\label{subsec:intro-cont}

This paper contributes to the connection between the topological
approach to language recognition and logical formalisms beyond the
setting of regular languages, and furthers, along the way, the study
of profinite monads in formal language theory. 

The main result of
Section~\ref{s:extending-set-bims} allows one to extend finitary
commutative $\Set$-monads to the category of \kl{Boolean spaces with
  internal monoids}. A particular instance of this result is presented
in Section~\ref{sec:measures-s}, where duality-theoretic insights are
used to provide a concrete and useful description of the
constructions involved in terms of \kl{measures}. In
Section~\ref{s:recog-transducers} we develop a generic approach for
mirroring operations on languages, such as modular quantifiers,
associating to a {\swim} $(X,M)$ a new {\swim} $(\DssX,\DssM)$. Finally,
   Section~\ref{sec:duality-expl} explains how
these constructions are indeed canonical and provides a
Reutenauer-type result characterising the Boolean algebra generated by 
the languages recognised by $(\DssX,\DssM)$.

\section{Preliminaries}
\label{sec:prelim}
\subsection{Logic on words}
\label{sec:logic-words}
\AP Fix an arbitrary finite set $A$, and write $\intro{\As}$ for the
free monoid over $A$.  A word over the alphabet $A$ (an \emph{$A$-word}, for short) is
an element $w\in \As$. In the logical approach to language theory, the word $w$ is regarded as a (relational) structure on the
set $\{1,\ldots, \len{w}\}$, where $\intro{\len{w}}$ denotes the
length of the word, equipped with a unary relation $P_a$ for each
$a\in A$ which singles out the positions in the word where the letter
$a$ appears.
If $\varphi$ is a sentence (i.e., a formula in which every variable is
in the scope of a quantifier) in a language interpretable over words,
we denote by $L_{\varphi}$ the set of words satisfying $\varphi$.

\AP Assume now $\varphi(x)$ is a formula with a free first-order
variable $x$ (intuitively this means that $\varphi(x)$ can talk about
positions in the word).  In order to be able to interpret the free
variable, we consider an extended alphabet $A\times 2$ which we think
of as consisting of two copies of $A$, that is, we identify $A\times
2$ with the set $A\cup\{a'\mid a\in A\}$, and we call the elements of
the second copy of $A$ \intro{marked letters}.  Assuming $w=a_1\ldots
a_n$ and $1\le i\leq\len{w}$, we write $\intro{\mw{i}}$ for the word
$a_1\ldots a_{i-1}a'_ia_{i+1}\ldots a_n$, i.e.\ for the word in $\Ats$
having the same shape as $w$ but with the letter in position $i$
marked, and $\intro{\wz}$ for the word $a_1\ldots a_n$ seen as a word
in $\Ats$.  Then we define $L_{\varphi(x)}$ as the set of all words in
the alphabet $A\times 2$ \emph{with only one \kl{marked letter}} such
that the underlying word in the alphabet $A$ satisfies $\varphi$ when
the variable $x$ points at the marked position.

\AP Now, given $L\subseteq \Ats$, denote by $\intro{\Lex}$ the
language consisting of those words $w=a_1\ldots a_n$ over $A$ such
that there exists $1\le i\le \len{w}$ with $a_1\ldots
a_{i-1}a'_ia_{i+1}\ldots a_n\in L$. Observe that $L=L_{\varphi(x)}$
entails $\Lex=L_{\exists x.\varphi(x)}$, thus recovering the usual
existential quantification.

\AP Among the generalisations of the existential quantifier are the
\intro{modular quantifiers}. Consider the ring $\intro*\Zq$ of
integers modulo $q$, and pick $p\in \Zq$. We say that an $A$-word $w$
satisfies the sentence $\modexists x.\varphi(x)$ if there exist $p$
modulo $q$ positions in $w$ for which the formula $\varphi(x)$
holds. Moreover, for an arbitrary language $L\subseteq \Ats$, we
define $\intro*\Lmodex$ as the set of $A$-words $w=a_1\ldots a_n$ such
that the cardinality of the set
\begin{align}\label{eq:set-of-witnesses}
  \{1\le i\le \len{w}\mid a_1\ldots a_{i-1}a'_ia_{i+1}\ldots a_n\in
  L\}
\end{align}
is congruent to $p$ modulo $q$. Clearly, if the language $L$ is
defined by the formula $\varphi(x)$, then $\Lmodex$ is defined by the
formula $\modexists x.\varphi(x)$.

\AP Finally, generalising the preceding situations, we can consider an
arbitrary \kl{semiring} $(S,+,\cdot,0_S,1_S)$ and an element $k\in S$. For
$L\subseteq \Ats$, an $A$-word $w=a_1\ldots a_n$ belongs to the
quantified language, denoted by $\intro*\Qk(L)$, provided that
\begin{align*}
  \underbrace{1_S+\cdots+ 1_S}_{m {\text{ times}}}=k,
\end{align*}
where $m$ is the cardinality of the set in
\eqref{eq:set-of-witnesses}.

\subsection{Stone duality and the \kl{Vietoris hyperspace}}
\AP Stone duality for Boolean algebras \cite{Stone1936} establishes a
categorical equivalence between the category of Boolean algebras 
and their homomorphisms, and the opposite of the
category $\intro*\Stone$ of \kl{Boolean (Stone) spaces} and continuous
maps between them.

\AP A \intro{Boolean space} is a compact Hausdorff space that admits a
basis of clopen (i.e., simultaneously closed and open) subsets.
There is an obvious forgetful functor $\fgt\colon\Stone\to\Set$. When
clear from the context, we will omit writing $\fgt$.

\AP The dual of the \kl{Boolean space} $X$ is the Boolean algebra
$\intro\Cl(X)$ of its clopen subsets, equipped with set-theoretic
operations. Conversely, given a Boolean algebra $B$, the dual space
$X$ may be taken either as the set of ultrafilters on $B$ (i.e., those
proper filters $F$ satisfying $a\in F$ or $\neg a\in F$ for every
$a\in B$) or as the Boolean algebra homomorphisms $h\colon B\to 2$
equipped with the topology generated by the sets
\begin{align*}
  \w{a}:=\{F\mid a\in F\}\cong\{h\mid h(a)=1\}, \ \text{for} \ a\in B.
\end{align*}

\AP An example of a \kl{Boolean space}, central to our treatment, is the
\intro{Stone-\v{C}ech compactification} of an arbitrary set $K$. This
is the dual space of the Boolean algebra $\P K$, and is denoted by
$\intro{\beta} K$. It is well known that the assignment $K\mapsto
\kl{\beta} K$ induces a functor $\kl{\beta}\colon \Set\to \Stone$ which
is left adjoint to the forgetful functor $\fgt\colon\Stone\to\Set$.
\AP Another functor, which played a key r\^ole in \cite{GehrkePR16}
and will serve here as a leading example, is the \intro{Vietoris
  functor} $\intro*\V\colon \Stone\to\Stone$. \AP Given a \kl{Boolean
  space} $X$, consider the collection $\kl{\V} X$ of all closed
subsets of $X$ equipped with the topology generated by the clopen
subbasis
\begin{align*}
  \{\Dv V\mid V\in \Cl(X)\}\cup\{(\Dv V)^c\mid V\in\Cl(X)\},
\end{align*}
where $\intro{\Dv} V:=\{K\in\V X \mid K\cap V\neq \emptyset\}.$
The resulting space is called the \emph{\kl{Vietoris
    \textup{(}hyper\textup{)}space}} of $X$, and is again a
\kl{Boolean space}. Further, if $f\colon X\to Y$ is a morphism in
$\Stone$, then so is the direct image function $\V X\to \V Y, \
K\mapsto f[K]$. In fact, it is well known that this is the functor
part of a monad $\V$ on $\Stone$.
The \kl{Vietoris hyperspace} of an arbitrary topological space was
first introduced by Vietoris~\cite{Vietoris1923}; for a complete account,
including results stated here without proof, see~\cite{Michael1951}.

\subsection{Boolean spaces with internal monoids}
In this section we give the definition of a \emph{\kl{Boolean space
    with an internal monoid}}, or {\swim} for short (see
Definition~\ref{def:swim} below), a topological recogniser well-suited
for dealing with non-regular languages.
In~\cite{GehrkePR16} a \emph{\kl{Boolean space with an internal
    monoid}} was defined as a pair $(X,M)$ consisting of a \kl{Boolean
  space} $X$, a dense subspace $M$ equipped with a monoid structure,
and a \AP\intro{biaction} (i.e., a pair of compatible left and right actions)
of $M$ on $X$ with continuous components extending the obvious
\kl{biaction} of $M$ on itself.
Here we use a small variation and simplification of this
notion. Instead of imposing that the monoid is a dense subset of the
space, we require a map from the monoid to the space with dense image.

\AP In what follows, for a \kl{Boolean space} $X$ we will denote by
$\intro{\cf{X}{X}}$ the set of continuous endofunctions on $X$, which
comes with the obvious monoid multiplication $\circ$ given by
composition. Given a monoid $(M,\cdot)$, we will denote by $r\colon
M\to M^M$ and $l\colon M\to M^M$ the two maps induced from the monoid
multiplication via currying, which correspond to the obvious right,
respectively left action of $M$ on itself.

\begin{definition}\label{def:swim}
  \AP A \intro{Boolean space with an internal monoid}, or a \swim, is
  a tuple $(X,M,h,\rho,\lambda)$, where $X$ is a \kl{Boolean space},
  $M$ is a monoid, $h\colon M\to X$, $\lambda\colon M\to\cf{X}{X}$ and
  $\rho\colon M\to\cf{X}{X}$ are functions such that $h$ has a dense
  image and for all $m\in M$ the following diagrams commute in $\Set$.
  \begin{equation}
    \label{eq:1}
    \begin{tikzcd}
      M\ar[r, "h"]\ar[d,"l(m)", swap] & X\ar[d,"\lambda(m)"] & & M\ar[r, "h"]\ar[d,"r(m)",swap] & X\ar[d,"\rho(m)"] \\
      M\ar[r, "h"] & X & & M\ar[r, "h"] & X
    \end{tikzcd}
  \end{equation}
  If no confusion arises we write $(X,M)$, or even just $X$, for the
  {\swim} $(X,M,h,\rho,\lambda)$.  \AP A \intro{morphism} between two
  {\swims} $X$ and $X'$ is a pair $(\tilde{\psi},\psi)$ where
  $\tilde{\psi}\colon X\to X'$ is a continuous map and $\psi\colon
  M\to M'$ is a monoid morphism such that $\tilde{\psi}\circ h=h'\circ\psi$. Note
  that since the image of $h$ is dense in $X$, given $\psi$,
  $\tilde{\psi}$ is uniquely determined if it exists. Accordingly, we
  will sometimes just write $\psi$ to designate the pair as well as
  each of its components.  We denote the ensuing category of {\swims}
  by $\intro{\swimcat}$.
\end{definition}

\begin{remark} 
Notice that if $(X,M)$ is a {\swim} of the form $(\kl{\beta}(\As),\As)$, and 
$(X',M')$ is any {\swim}, then every monoid morphism $\psi\colon\As\to M'$
yields a (unique) continuous extension $\tilde{\psi}\colon\kl{\beta}(\As)\to X'$
making the pair $(\tilde{\psi},\psi)$ into a {\swim} \kl{morphism}. Thus {\swim} 
\kl{morphisms} $(\kl{\beta}(\As),\As)\to(X',M')$ are in one-to-one correspondence 
with monoid morphisms $\As\to M'$. For this reason we will often treat these two things as one and the same.
\end{remark}

\begin{remark} From Definition \ref{def:swim} it follows that $\rho$ and
  $\lambda$ induce in fact commuting right and left $M$-actions on
  $X$, so that $h$ is an $M$-\kl{biaction} morphism. Indeed, since $h$
  has a dense image in $X$ it follows that $\rho(m)$ and $\lambda(m)$
  are the unique extensions on $X$ of $r(m)$, respectively $l(m)$. But
  the left and right actions of $M$ on itself commute, hence $\rho$
  and $\lambda$ must enjoy the same properties.
  We also obtain that $(X, \Im(h))$ is a Boolean space with an
  internal monoid exactly as defined in~\cite{GehrkePR16}.
\end{remark}
\begin{remark}\label{rem:reformulation-bim-def}
  An equivalent way of saying that the diagrams in~\eqref{eq:1} commute for
  all $m\in M$ is to say that the following diagrams commute in
  $\Set$.
  \begin{equation*}
    \begin{tikzcd}
      {\cf{X}{X}}\ar[r,"-\circ h"] & X^M  & {\cf{X}{X}}\ar[r,"-\circ h"] & X^M \\
      M\ar[u,"\lambda"]\ar[r,"l"] & M^M\ar[u,"h\circ -", swap] &
      M\ar[u,"\rho"]\ar[r,"r"] & M^M\ar[u,"h\circ -", swap]
    \end{tikzcd}
  \end{equation*}
  This will come handy in the proof of Theorem~\ref{thm:lift}.
\end{remark}
To conclude, we recall the associated notion of recognition. Under the
bijection between subsets of a given set $K$ and clopens of its
\kl{Stone-\v{C}ech compactification} $\kl{\beta} K$, we write $\w{L}$
for the clopen associated with the subset $L\in \P K$.
\begin{definition}\label{def:recognition-via-bims}
  \AP Let $A$ be a finite alphabet and $L\in \P(\As)$. A
  \kl{morphism} of {\swims} $\psi\colon (\kl{\beta}(\As),\As)\to(X,M)$
  \intro{recognises} the language $L$ if there is a clopen $C\subseteq
  X$ such that $\psi^{-1}(C)=\widehat{L}$.
  Moreover, we say that the {\swim} $(X,M)$ recognises the language
  $L$ if there exists a {\swim} \kl{morphism}
  $(\kl{\beta}(\As),\As)\to(X,M)$ \kl{recognising} $L$.
  Finally, if $\B\hookrightarrow \P(\As)$ is a Boolean subalgebra, the
  {\swim} $(X,M)$ is said to recognise $\B$ provided that it
  \kl{recognises} each $L\in \B$.
\end{definition}

Equivalently, a language $L\in \P(\As)$ is recognised by the morphism
of {\swims} $\psi\colon (\kl{\beta}(\As),\As)\to(X,M)$
when there exists a clopen $C\subseteq X$ such that
$L=\psi^{-1}(h^{-1}(C))$.  The topology on $X$ specifies which subsets
of $M$ can be used for recognition, namely the preimages under $h$ of
the clopens in $X$. However, when $M$ is finite so is $X$. In fact, in
this case $X$ has the same carrier set as $M$ and is equipped with the
discrete topology, therefore in the regular setting we recover the usual
notion of recognition.

\subsection{Monads and algebras}
We assume the reader is familiar with the basic notions of category
theory, and especially with monads as a categorical approach to
general algebra. Concerning the latter, we refer the reader to, e.g.,
\cite[Chapter~VI]{MacLane} or~\cite[Chapters~3--4]{Borceux2}.

\AP Consider a monad $(T,\eta,\mu)$ on a category $\C$. Recall that an
\intro{Eilenberg-Moore algebra} for $T$ (or a
\emph{\kl{$T$-algebra}}, for short) is a pair $(X,h)$ where $X$ is an
object of $\C$ and $h\colon TX\to X$ is a morphism in $\C$ which
behaves well with respect to the unit $\eta$ and the multiplication
$\mu$ of the monad, that is, $h\circ \eta_X=\id_X$ and $h\circ Th =
h\circ \mu_X$.
A \emph{morphism} of $T$-algebras $(X_1,h_1)\to(X_2,h_2)$ is a
morphism $f\colon X_1\to X_2$ in $\C$ satisfying $f\circ h_1=h_2\circ
Tf$.
Let $\intro*\AlgCT$ denote the category of \kl{$T$-algebras}.
When $T$ is a monad on the category $\Set$ of sets and functions,
categories of the form $\intro*\AlgT$ are, up to equivalence,
precisely the varieties of (possibly infinite arity) algebras. This
correspondence restricts to categories of \kl{Eilenberg-Moore
  algebras} for \intro{finitary} monads (i.e., monads preserving filtered
colimits) and varieties of algebras in types consisting of finite
arity operations.
\AP A $T$-algebra $(X,h)$ is said to be \emph{finitely carried} (or
sometimes just \emph{finite}) provided $X$ is finite. We write
$\intro*\AlgfinT$ for the full subcategory of $\AlgT$ on the finitely
carried objects. \AP The forgetful functor $\AlgT\to\Set$ that sends
$(X,h)$ to $X$ restricts to the finitely carried algebras, and gives
rise to a functor $\AlgfinT\to\Setfin$.

\AP In Section~\ref{sec:S-transduction} we shall see how several
logical quantifiers can be modelled by considering modules over a
semiring and the appropriate profinite monad. Recall that a
\intro{semiring} is a tuple $(S,+,\cdot,0,1)$ such that $(S,+,0)$ is a
commutative monoid, $(S,\cdot, 1)$ is a monoid, the operation $\cdot$
distributes over $+$, and $0\cdot s=0=s\cdot 0$ for all $s\in S$. If
no confusion arises, we will denote the \kl{semiring} by $S$ only.
\begin{example}
  \label{ex:semiring-monad}
  \AP A \kl{semiring} $S$ induces a functor
  $\intro\Sm\colon\Set\to\Set$ which associates with a set $X$ the set
  of all functions $X\to S$ with \emph{finite support}, that is
  \begin{align*}
    \Sm X := \{f\colon X \to S \mid f(x) = 0 \ \text{for all but
      finitely many} \ x\in X\}.
  \end{align*}
  If $\psi\colon X\to Y$ is any function, define $ \Sm \psi\colon \Sm
  X\to\Sm Y$ as $f\mapsto (y\mapsto\sum_{\psi(x)=y}{f(x)})$.
  Any element $f\in \Sm X$ can be represented as a formal sum
  $\sum_{i=1}^n{s_ix_i}$, where $\{x_1,\ldots,x_n\}$ is the support of
  $f$ and $s_i=f(x_i)$ for each $i$.
  The functor $\Sm$ is part of a monad $(\Sm,\eta,\mu)$ on $\Set$,
  called the \intro{free $S$-semimodule monad}, whose unit is
  \begin{align*}
    \eta_X\colon X\to \Sm X, \ \eta_X(x)(x')=1 \ \text{if} \ x'=x \
    \text{and} \ 0 \ \text{otherwise},
  \end{align*}
  and whose multiplication is
  \begin{align*}
    \mu_X\colon \Sm^2 X\to \Sm X, \ \ \sum_{i=1}^n{s_i f_i}\mapsto
    \bigg(x\mapsto \sum_{i=1}^n{s_i f_i(x)}\bigg).
  \end{align*}
  The category $\Alg{\Set}{\Sm}$ is the category of modules over the
  \kl{semiring} $S$.
For example, if $S$ is the Boolean semiring $\2$ then $\Sm=\Pfin$ (the
  finite powerset monad), whose \kl{Eilenberg-Moore algebras} are join
  semilattices. If $S$ is $(\N,+,\cdot,0,1)$ or $(\Z,+,\cdot,0,1)$,
  then the algebras for the monad $\Sm$ are, respectively, Abelian
  monoids and Abelian groups.
\end{example}

\subsection{Profinite monads}
\label{sec:profinite-monads}

Throughout this subsection we fix a monad $T$ on $\Set$. We begin by
recalling the definition of the associated \emph{profinite monad}
$\intro{\wT}$ on the category of \kl{Boolean spaces},
following~\cite{ChenAMU16}.  First we provide an intuitive idea of the
construction, and then we give the formal definition. Given a
\kl{Boolean space} $X$, one considers all continuous maps $h_i\colon
X\to Y_i$ where the $Y_i$'s are finite sets equipped with
Eilenberg-Moore algebra structures $\alpha_i\colon TY_i\to Y_i$, as
well as the algebra morphisms $u_{ij}\colon Y_i\to Y_j$ satisfying
$u_{ij}\circ h_i=h_j$. Equipping the finite sets $Y_i$ with the
discrete topology, one obtains a cofiltered diagram (or inverse
limit system) $\mathcal{D}_X$ in $\Stone$, and ${\wT}X$ is
the limit of this system. It turns out that ${\wT}$ is the underlying functor of a
monad $({\wT},\w{\eta},\w{\mu})$ on $\Stone$, called the
\emph{profinite monad associated with $T$}.  For example, it is not difficult to see how to obtain its unit $\w{\eta}_X$ from the universal
property of the limit, as in the following diagram.
\[
\begin{tikzcd}[row sep =small]
  X\ar[dr, swap, "h_i"]\arrow[bend left]{drr}[swap]{h_j}\ar[rrr,bend left,dashed,"\w{\eta}_X"] & &  & {\wT}X\ar[dll, bend right, "p_i"]\ar[dl,"p_j"] \\
  {} & Y_i\ar[r,"u_{ij}"] & Y_j &
\end{tikzcd}
\]
To give the formal definition of ${\wT}$, we introduce the functor
$G\colon\AlgfinT\to\Stone$ obtained as the composition of the
forgetful functor to $\Setfin$ with the embedding of $\Setfin$ into
$\Stone$:
\begin{equation*}
  \begin{tikzcd}
    G\colon \AlgfinT\ar[r]&\Setfin\ar[r]&\Stone.
  \end{tikzcd}
\end{equation*}
The shape of the diagram we constructed above for a \kl{Boolean space}
$X$ is the comma category $X\downarrow G$ whose objects are
essentially the maps $h_i\colon X\to G(Y_i,\alpha_i)$, and whose
arrows are the maps $u_{ij}$ as above. The diagram $\mathcal{D}_X$
is then formally given by the composition
\[
\begin{tikzcd}
  X\downarrow G\ar[r,"cod"] &\AlgfinT\ar[r,"G"]& \Stone
\end{tikzcd}
\]
of the codomain functor $ X\downarrow G \to \AlgfinT$, which maps
$h_i\colon X\to G(Y_i,\alpha_i)$ to the algebra $(Y_i,\alpha_i)$, and
$G\colon\AlgfinT\to\Stone$.  Formally, for an arbitrary \kl{Boolean
  space} $X$, we have ${\wT}X:=\lim\mathcal{D}_X$.

\AP Notice that this is the pointwise limit computation of the
right Kan extension of $G$ along itself (cf.\ \cite[X.3]{MacLane}).
That is, using standard category-theoretic notation, ${\wT}=\Ran_GG$.
It is well known (see for example~\cite{Leinster2010}) that the right
Kan extension of a functor $G$ along itself, when it exists, is the
functor part of a monad, called the \intro{codensity monad} for $G$.

\AP The universal property of the right Kan extension, along with the
fact that the underlying-set functor $\fgt\colon\Stone\to\Set$ is
right adjoint and thus preserves right Kan extensions, allows one to
define a natural transformation
\begin{align}
  \label{eq:4}
  \t_X\colon T|X|\to |{\wT}X|
\end{align}
which was also used in~\cite{ChenAMU16}.  Here we give a presentation
based on the limit computation of ${\wT}X$.  Notice that the maps
$|h_i|\colon |X|\to |Y_i|$ are functions into the carrier sets of the
\kl{Eilenberg-Moore algebras} $\alpha_i\colon TY_i\to Y_i$ and thus,
by the universal property of the free algebra $T|X|$, we can extend
the maps $|h_i|$ to algebra morphisms $h_i^\#$ from $T|X|$ to
$(Y_i,\alpha_i)$. The functions $h_i^\#$ form a cone for the diagram
$\fgt\circ\mathcal{D}_X$ in $\Set$ whose limit is $|{\wT}X|$, by
virtue of the fact that the forgetful functor
$\fgt\colon\Stone\to\Set$ preserves limits. By the universal property
of the limit, this yields a unique map $\t_X$ as in~\eqref{eq:4}.

The natural transformation $\t$ behaves well with respect to the units
and multiplications of the monads $T$ and $\wT$, in the sense that the
next two diagrams commute,
see~\cite[Proposition~B.7]{DBLP:journals/corr/ChenAMU15}. Thus the
pair $(\fgt,\t)$ is a \emph{monad morphism}, or \emph{monad functor}
in the terminology of~\cite{Street}.
\begin{equation}
  \label{eq:tau-monad-functor}
  \begin{tikzcd}[column sep=small]
    T|X| \arrow{rr}{\t_X} &  & {|{\wT}X|} & & T^2|X| \arrow{d}[swap]{T\t_X} \arrow{rr}{\mu_{|X|}} & & T|X| \arrow{d}{\t_X} \\
    {} & {|X|}\arrow{ur}[swap]{|\w{\eta}_X|} \arrow{ul}{\eta_{|X|}} &
    & & T|{\wT}X| \arrow{r}[swap]{\t_{{\wT}X}} & {|{\wT}^2X|}
    \arrow{r}[swap]{|\w{\mu}_X|} & {|{\wT}X|}
  \end{tikzcd}
\end{equation}

\AP The fact that $(\fgt,\t)$ is a monad functor entails that the
functor $\fgt$ lifts to a functor $\intro*\wfgt$ between the
categories of \kl{Eilenberg-Moore algebras} for the monads ${\wT}$ and
$T$, as in the next diagram.
\begin{equation}
  \label{eq:8}
  \begin{tikzcd}
    \Stone^{{\wT}}\ar[r,"\wfgt"]\ar[d] & \AlgT\ar[d]\\
    \Stone\ar[r,"\fgt"] & \Set
  \end{tikzcd}
\end{equation}
As a consequence we immediately obtain that the set $|{\wT}X|$ admits
a \kl{$T$-algebra} structure, a result also used
in~\cite{DBLP:journals/corr/ChenAMU15} for finite algebras. This
structure is essentially the one obtained by applying the functor
$\wfgt$ to the free ${\wT}$-algebra $({\wT}X,\w{\mu}_X)$. In more
detail,
\begin{lemma}
  \label{lem:codensity-T-alg}
  Given a \kl{Boolean space} $X$, the composite map
  \[
  \begin{tikzcd}
    T{|{\wT}X|}\ar[r,"\t_{{\wT}X}"]& {|{\wT}^2X|}\ar[r,"|\w{\mu}_X|"]
    & {|{\wT}X|}
  \end{tikzcd}
  \]
  is a \kl{$T$-algebra} structure on $|{\wT}X|$. Moreover, $\t_X$ is a
  morphism of $T$-algebras from the free $T$-algebra on $|X|$ to $|\wT
  X|$ with the above structure.
\end{lemma}
\begin{proof}
  This is a straightforward verification using the commutativity of
  the diagrams in \eqref{eq:tau-monad-functor}.
\end{proof}
While in some proofs it is essential to keep track of the forgetful
functor, we will sometimes omit it in what follows and simply write
$\t_X\colon TX\to{\wT}X$. We recall a property of the natural transformation 
$\t$ which will be crucial in the following.
\begin{lemma}\label{l:t-dense-injective}
  For every \kl{Boolean space} $X$, the map $\t_X\colon TX\to {\wT}X$ has dense image. 
  More generally, the composite
  \[TM\xrightarrow{Th} TX\xrightarrow{\t_X} {\wT}X\] 
  has dense image whenever $h\colon M\to X$ is a function with dense image.
\end{lemma}
\begin{proof}
For a proof of the fact that $\t_X$ has dense image, see \cite[Lemma~2.9]{R2018}. An easy adaptation of the latter proof yields the second part of the statement.
\end{proof}

\begin{remark}\label{rm:injective-components-of-t}
  Notice that, for an arbitrary monad $T$ on $\Set$, the components of the natural transformation $\t$ from \eqref{eq:4} 
  do not have to be injective. A counterexample is provided by the powerset monad $\P$ on $\Set$. Indeed, both 
  $\P$ and $\Pfin$ generate the same profinite monad, namely the \kl{Vietoris monad} on $\Stone$. In the case of the monad $\P$, $\t_X\colon \P X\to \V X$ sends a subset of the \kl{Boolean space} $X$ to its closure, and this function is not injective in general.
  However, the components of $\t$ are injective if $T$ is \kl{finitary} and restricts to finite sets, e.g., if $T$ is the finite powerset monad on $\Set$.  For more details we 
  refer the reader to \cite[Section~2.2]{R2018}.
\end{remark}

\section{Extending $\Set$-monads to \swims}\label{s:extending-set-bims}
In this section we study liftings of monads from the category of sets to the category of {\swims}. Let us fix, throughout the section, a monad $T$ on $\Set$. In Section~\ref{sec:profinite-monads} we have seen
that the profinite monad $\wT$ provides a canonical
way of extending $T$ to \kl{Boolean spaces}.  On the other hand, in Section \ref{s:lifting-to-Mon} we consider
ways of lifting $T$ to the category of monoids. The combination of these two liftings, the topological and the monoid one, is considered in Section \ref{s:combining-liftings}. In particular, in Theorem~\ref{thm:lift} we give sufficient conditions for $T$ to be extended in a canonical way to the category of {\swims}, by combining the aforementioned liftings. 

\subsection{Lifting $\Set$-monads to the category of monoids}\label{s:lifting-to-Mon}
It is well known that there are two ``canonical'' natural
transformations of bifunctors 
\[
\otimes,\otimes'\colon TX\times TY \to T(X\times Y),
\] 
defined intuitively as follows. If we think of elements in
$TX$ as terms $t(x_1,\ldots, x_n)$, then $t(x_1,\ldots, x_n) \otimes
s(y_1,\ldots, y_m)$ is defined as
\[
t(s((x_1,y_1),\ldots, (x_1,y_m)),\ldots,s((x_n,y_1),\ldots,
(x_n,y_m))),
\]
whereas $t(x_1,\ldots, x_n) \otimes' s(y_1,\ldots, y_m)$ is defined as
\[
s(t((x_1,y_1),\ldots, (x_n,y_1)),\ldots,t((x_1,y_m),\ldots,
(x_n,y_m))).
\]
\AP In general $\otimes$ and $\otimes'$ do not coincide, and when they
do the monad is called \kl{commutative}, a notion due to
Kock~\cite{Kock70}. We give a formal definition in the case of the monad $T$. Every $\Set$-monad has a unique \emph{strength}, that is a natural transformation $\sigma_{X,Y}\colon X\times TY\to T(X\times Y)$ such that the following diagrams commute.
\begin{equation}\label{eq:strength}
\begin{tikzcd}[column sep=3em]
X\times Y \arrow{r}{\id_X\times \eta_Y} \arrow{dr}[swap]{\eta_{X\times Y}} & X\times TY \arrow{d}{\sigma_{X,Y}} & & X\times T^2Y \arrow{r}{\sigma_{X,TY}} \arrow{d}[swap]{\id_X\times\mu_Y} & T(X\times TY) \arrow{r}{T\sigma_{X,Y}} & T^2(X\times Y) \arrow{d}{\mu_{X\times Y}} \\ 
{} & T(X\times Y) & & X\times TY \arrow{rr}{\sigma_{X,Y}} & & T(X\times Y)
\end{tikzcd}
\end{equation}
This natural transformation can be explicitly described as follows. For any $x\in X$, write $f_x\colon Y\to X\times Y$ for the function sending $y$ to $(x,y)$. Then $\sigma_{X,Y}\colon X\times TY\to T(X\times Y)$ sends a pair $(x,s)$ to the image of $s$ under $Tf_x\colon TY\to T(X\times Y)$. Associated with the strength $\sigma$, there is a \emph{costrength} $\sigma'_{X,Y}\colon TX\times Y\to T(X\times Y)$ defined as the composition
\[\begin{tikzcd}
TX\times Y \arrow{r}{\gamma_{TX,Y}} & Y\times TX \arrow{r}{\sigma_{Y,X}} & T(Y\times X) \arrow{r}{T\gamma_{Y,X}} & T(X\times Y),
\end{tikzcd}\]
where $\gamma_{X,Y}\colon X\times Y \to Y\times X$ is the function sending $(x,y)$ to $(y,x)$.
The costrength $\sigma'$ enjoys properties symmetric to those of the strength $\sigma$, expressed by the following commutative diagrams.
\begin{equation}\label{eq:costrength}
\begin{tikzcd}[column sep=3em]
X\times Y \arrow{r}{\eta_X\times \id_Y} \arrow{dr}[swap]{\eta_{X\times Y}} & TX\times Y \arrow{d}{\sigma'_{X,Y}} & & T^2X\times Y \arrow{r}{\sigma'_{TX,Y}} \arrow{d}[swap]{\mu_X\times\id_Y} & T(TX\times Y) \arrow{r}{T\sigma'_{X,Y}} & T^2(X\times Y) \arrow{d}{\mu_{X\times Y}} \\ 
{} & T(X\times Y) & & TX\times Y \arrow{rr}{\sigma'_{X,Y}} & & T(X\times Y)
\end{tikzcd}
\end{equation}
\AP The monad $T$ is said to be \intro{commutative} if, for any sets $X,Y$, the following square commutes.
\begin{equation}\label{eq:commutative-monad}
\begin{tikzcd}
TX\times TY \arrow{r}{\sigma_{TX,Y}} \arrow{d}[swap]{\sigma'_{X,TY}} & T(TX\times Y) \arrow{r}{T\sigma'_{X,Y}} & T^2(X\times Y) \arrow{dd}{\mu_{X\times Y}} \\
T(X\times TY) \arrow{d}[swap]{T\sigma_{X,Y}} & & \\
T^2(X\times Y) \arrow{rr}{\mu_{X\times Y}} & & T(X\times Y)
\end{tikzcd}
\end{equation}
Note that the commutativity of this diagram formalises the aforementioned idea that the natural transformations $\otimes$ and $\otimes'$ coincide.
Given a monoid $(M,\cdot,1)$, one has two possibly different
``canonical'' ways of defining a binary operation on $TM$, obtained as
either of the two composites
\begin{equation}\label{eq:monoid-operation-lifting}
\begin{tikzcd}
  TM\times TM\arrow[shift right]{r}[swap]{\otimes'}\ar[r,shift
  left,"\otimes"] & T(M\times M)\ar[r,"T(\cdot)"] & TM.
\end{tikzcd}
\end{equation}
If $e\colon 1\to M$ denotes the map selecting the unit of the monoid,
we can also define a map $1\to TM$ obtained as the composite
$Te\circ\eta_1$. That these data (with either of the two binary
operations) give rise to monoid structures on $TM$ is a consequence
of a more general result by Kock:
\begin{theorem}
If $T$ is a commutative $\Set$-monad then $\otimes=\otimes'$, and thus for every monoid $(M,\cdot,1)$ the composition in \eqref{eq:monoid-operation-lifting} gives a monoid structure on $TM$. This yields a lifting of $T$ to a monad on the category of monoids and monoid homomorphisms.
\end{theorem}
\begin{proof}
This is a special case of~\cite[Theorem~2.1]{Kock70}.
\end{proof}

\subsection{Combining the topological and monoid liftings}\label{s:combining-liftings}
In Sections \ref{sec:profinite-monads} and \ref{s:lifting-to-Mon}, respectively, we have seen that every $\Set$-monad $T$ can be lifted to a monad ${\wT}$ on the category of Boolean spaces, and it can be lifted to a monad on the category of monoids provided it is commutative.
In this section we show that, if $T$ is commutative and finitary, then the topological and monoid liftings can be combined to obtain a lifting of $T$ to the category of {\swims} (see Theorem \ref{thm:lift} below).

\AP Let $T$ be a commutative $\Set$-monad, $(X,M)$ a {\swim} and $h\colon M\to X$ the associated function with dense image. We would like to define a structure of {\swim} on the pair $({\wT}X,TM)$. In particular, we should give a function $TM\to {\wT}X$ with dense image. To this aim, we define $\intro*{\wh}\colon TM\to {\wT}X$ as the composition
\begin{equation}\label{eq:map-h-hat}
\begin{tikzcd}
    TM \arrow{r}{Th} & TX \arrow{r}{\t_X} & {\wT}X.
    \end{tikzcd}
  \end{equation}
  By Lemma~\ref{l:t-dense-injective}, this function has
  dense image. Further, since both $Th$ and $\t_X$ are $T$-algebra morphisms,
  $\wh$ is also a $T$-algebra morphism. 
  
  Recall that, if
  $\alpha\colon TB\to B$ is a \kl{$T$-algebra} and $A$ is any set,
  the set of functions ${B^A}$ carries an \kl{Eilenberg-Moore algebra}
  structure for $T$, where the operations are defined pointwise. Further, whenever
  $\alpha_i\colon TB_i\to B_i$ for $i\in\{1,2\}$ are
  \kl{Eilenberg-Moore algebras} for $T$ and $f\colon B_1\to B_2$ is an
  algebra morphism, then $\Set(A,f)=f\circ -\colon{B_1^A}\to{B_2^A}$ is a
  \kl{$T$-algebra} morphism.
  We obtain at once the following fact.
\begin{lemma}
  \label{lem:pointwise-T-algebras}
  For any $\Set$-monad $T$, the sets $TM^{TM}$ and ${\wT}X^{TM}$ carry \kl{$T$-algebra} structures and the function 
  \[
  \wh\circ-\colon TM^{TM}\to{\wT}X^{TM}
  \]
  is a \kl{$T$-algebra} morphism.
\end{lemma}
\begin{proof}
With the notation of the previous paragraph, consider $A=TM$, $\alpha_1=\mu_M\colon T^2M\to
  TM$, $\alpha_2$ the $T$-algebra structure on ${\wT}X$ given as
  in Lemma~\ref{lem:codensity-T-alg}, and $f=\wh$.
\end{proof}
Thus, also the power algebra ${\wT X}^{\wT X}$ admits a $T$-algebra structure. Crucially, if the monad $T$ is finitary, the set $\cf{\wT X}{\wT X}$ of continuous endofunctions on $\wT X$ is a subalgebra of ${\wT X}^{\wT X}$. This is proved in the following proposition which will allow us to define, in the proof of Theorem \ref{thm:lift}, a biaction of $TM$ on $\wT X$.
\begin{proposition}\label{p:finitary-subalgebra-continuous-functions}
If $T$ is a finitary $\Set$-monad, then $\cf{\wT X}{\wT X}$ is a subalgebra of the $T$-algebra ${\wT X}^{\wT X}$. With respect to this structure, the function 
  \[
  -\circ\wh\colon \cf{\wT X}{\wT X}\to{\wT}X^{TM}
  \]
  is a $T$-algebra morphism.
\end{proposition}
\begin{proof}
  It suffices to prove the first part of the statement, for then the function $-\circ\wh$ is a $T$-algebra morphism because it coincides with the following composition of $T$-algebra morphisms:
  \[\begin{tikzcd}
  \cf{\wT X}{\wT X} \arrow[hookrightarrow]{r} & {\wT X}^{\wT X} \arrow{r}{{\wT X}^{\wh}} & {\wT X}^{TM}.
  \end{tikzcd}\] 
  
  Recall from
  Section~\ref{sec:profinite-monads} that ${\wT}X$ is the cofiltered
  limit of finite sets $Y_i$ which carry \kl{$T$-algebra} structures
  $\alpha_i\colon TY_i\to Y_i$. We have the following isomorphisms in
  the category of sets:
  \begin{align*}
    \cf{\wT X}{\wT X}& \cong \cf{\wT X}{\mathrm{lim}_iY_i}\\
    & \cong \mathrm{lim}_i\cf{\wT X}{Y_i}\\
    & \cong \mathrm{lim}_i\cf{\mathrm{lim}_jY_j}{Y_i}\\
    & \cong \mathrm{lim}_i\colim_j\cf{Y_j}{Y_i}
  \end{align*}
  where, for the last isomorphism, we have used the fact that the
  $Y_i$ are finite spaces, and consequently finitely
  copresentable. Moreover, notice that the colimit above is filtered.

  The sets $Y_i$ carry
  $T$-algebra structures and so do the sets $\cf{Y_j}{Y_i}\cong
  Y_i^{Y_j}$ with respect to pointwise operations.
  Since $T$ is finitary, the forgetful functor $\AlgT\to\Set$
  preserves and reflects both filtered colimits and limits (see
  e.g.~\cite[Propositions~3.4.1--3.4.2]{Borceux2}). Whence,
  $\cf{\wT X}{\wT X}$ carries a \kl{$T$-algebra} structure.
  We claim that, with respect to this $T$-algebra structure, $\cf{\wT X}{\wT X}$ is a subalgebra
  of the power algebra ${\wT X}^{\wT X}$. 

   For each $x\in \wT X$, write $ev_x\colon  \cf{\wT X}{\wT X}\to \wT X$ for the function sending $f$ to $f(x)$. Then the natural inclusion 
   \[
   \cf{\wT X}{\wT X} \hookrightarrow {\wT X}^{\wT X}
   \] 
   is a $T$-algebra morphism if, and only if, each $ev_x$ is a $T$-algebra morphism. Write $\{\pi_i\colon \wT X\to Y_i\mid i\in I\}$ for the cone of continuous functions defining $\wT X$ as the cofiltered limit of finite sets $Y_i$ which carry $T$-algebra structures. It is not difficult to see that each $\pi_i$ is a $T$-algebra morphism; for a proof, see \cite[Proposition~2.10]{R2018}. Therefore it suffices to show that each composition  $\pi_i\circ ev_x\colon \cf{\wT X}{\wT X}\to Y_i$ is a $T$-algebra morphism. For any $j\in I$, denote by $\gamma_j\colon \wT X^{Y_j}\to Y_i$ the composite
   \[\begin{tikzcd}[column sep=3em]
   \wT X^{Y_j} \arrow{r}{\pi_i\circ -} & Y_i^{Y_j} \arrow{r}{ev_{\pi_j(x)}} & Y_i.
   \end{tikzcd}\]
 It follows by Lemma \ref{lem:pointwise-T-algebras} that each $\gamma_j$ is a $T$-algebra morphism. Upon recalling that $\cf{\wT X}{\wT X}\cong \colim_j\cf{Y_j}{\wT X}$ in the category $\AlgT$, it is not difficult to see that $\pi_i\circ ev_x\colon \cf{\wT X}{\wT X}\to Y_i$ is the (unique) $T$-algebra morphism induced by the cocone $\{\gamma_j\colon \cf{Y_j}{\wT X}\to Y_i\mid j\in I\}$, thus concluding the proof.
\end{proof}
Exploiting the previous observations we can prove the main result of this section, stating that every finitary commutative monad on the category of sets can be lifted to the category of {\swims}.
\begin{theorem}\label{thm:lift}
  Any \kl{finitary} \kl{commutative} $\Set$-monad $T$ can be extended
  to a monad on {\swimcat} mapping $(X,M)$ to $({\wT}X,TM)$.
\end{theorem}
\begin{proof} We first give the definition of the monad on an object
  $(X,M,h,\rho,\lambda)$. We will show that this is mapped to a
  {\swim} $({\wT}X,TM,\intro*{\wh},\w{\rho},\w{\lambda})$, where
  $\wh$ is as in equation \eqref{eq:map-h-hat}, and $\w{\rho}$ and $\w{\lambda}$ are defined as follows. 
  \AP Recall that $\cf{X}{X}$ and $\cf{\wT X}{\wT X}$ denote the sets
  of continuous endofunctions on $X$ and $\wT X$, respectively. To define
  $\w{\rho}$, consider the composite of the following two maps, where
  $\intro{\Ac{\wT}{X}}$ is given by the application of the functor
  ${\wT}$ to a continuous function in $\cf{X}{X}$:
  \begin{equation}
    \label{eq:3}
    \begin{tikzcd}
      M\arrow{r}{\rho} & {\cf{X}{X}} \ar{r}{\Ac{\wT}{X}} & {\cf{\wT X}{\wT X}}.
    \end{tikzcd}
  \end{equation}
  By Proposition \ref{p:finitary-subalgebra-continuous-functions} we know that $\cf{\wT X}{\wT X}$ is a $T$-algebra, hence the map in~\eqref{eq:3}  admits a unique extension to an algebra morphism
  $\w{\rho}\colon TM\to \cf{\wT X}{\wT X}$.
  The function $\w{\lambda}$ is defined similarly, as the unique
  \kl{$T$-algebra} morphism extending $\Ac{\w T}{X}\circ\lambda$.

  In order to prove that
  $({\wT}X,TM,\intro*{\wh},\w{\rho},\w{\lambda})$ is a {\swim}, it
  remains to prove that the functions $\wh$, $\w{\rho}$ and
  $\w{\lambda}$ make the diagrams in Definition~\ref{def:swim}
  commute. Equivalently, by virtue of
  Remark~\ref{rem:reformulation-bim-def}, that the next square and the
  analogous one (with $\w{\rho}$ replaced by $\w{\lambda}$, and
  $\w{r}$ by $\w{l}$) commute,
  \begin{equation}
    \label{eq:5}
    \begin{tikzcd}
      {\cf{\wT X}{\wT X}}\ar[r,"-\circ \wh"] & {{\wT}X^{TM}}
      \\
      TM\ar[u,"\w{\rho}"]\ar[r,"\w{r}"] & {TM^{TM}}\ar[u,"\wh\circ
      -",swap]
    \end{tikzcd}
  \end{equation}
  where $\w{r}$ and $\w{l}$ denote the right and left action,
  respectively, of $TM$ on itself.
  To this end, notice that the following diagram commutes. 
  \begin{equation}
  \label{eq:two-trapezoids}
  \begin{tikzcd}
    {\cf{\wT X}{\wT X}}\ar[rr,"-\circ \wh"]& & {\wT}X^{TM} \\
    {\cf{X}{X}}\ar[u,"\Ac{\wT}{X}"]\ar[r,"-\circ h"] & X^M\ar[ru,"\t_X\circ T-"] & \\
    M\ar[u,"\rho"]\ar[r,"r"] & M^M\ar[u,"h\circ -",
    swap]\ar[r,"\Ac{T}{M}"] & TM^{TM}\ar[uu,"\wh\circ -",swap]
  \end{tikzcd}
  \end{equation}
  For the upper leftmost trapezoid, by definition of $\wh$, we must prove that for every $f\in\cf{X}{X}$ we have 
  \[
  \t_X\circ Tf\circ Th=  \t_X\circ T(f\circ h)=\wT f\circ \wh =\wT f\circ \t_X\circ Th.
  \] 
  In turn, this follows from the fact that $\t_X\circ Tf=\wT f\circ \t_X$ by naturality of $\t$.
  The lower rightmost trapezoid commutes by the very definition of
  $\wh$, whereas the inner square is a reformulation of the left
  commuting square in~\eqref{eq:1}, cf.\ Remark \ref{rem:reformulation-bim-def}.
  
  We derive the commutativity of~\eqref{eq:5} using the universal
  property of the free \kl{$T$-algebra} on $M$ and by observing that
  a) in the outer square in~\eqref{eq:two-trapezoids}, the right
  vertical and the top horizontal arrows are morphisms of \kl{$T$-algebras} by
  Lemma~\ref{lem:pointwise-T-algebras} and Proposition \ref{p:finitary-subalgebra-continuous-functions}, respectively; b) the map $\w{\rho}$
  was defined as the unique extension of $\Ac{\wT}{X}\circ \rho$ to
  the free algebra $TM$; c) the map $\w{r}$ is the unique algebra
  morphism extending $\Ac{T}{M}\circ r$ to $TM$.
  To settle item c), notice that it is equivalent to the commutativity of the following diagram,
  \begin{equation}\label{eq-item-c}
  \begin{tikzcd}[column sep=4em]
  TM \times M \arrow{r}{\id_{TM}\times \eta_M} \arrow{dd}[swap]{\id_{TM}\times r} & TM\times TM \arrow{d}{\otimes}\\
  {} & T(M\times M) \arrow{d}{T(\cdot)} \\
  TM\times M^M \arrow{r}{\epsilon} & TM
  \end{tikzcd}
  \end{equation}
  where $\otimes$ denotes either of the two compositions in diagram \eqref{eq:commutative-monad}, $\cdot\, \colon M\times M\to M$ is the monoid operation of $M$ and $\epsilon(s,f)=\Ac{T}{M}(f)(s)$ for every $(s,f)\in TM\times M^M$. Now, observe that the identity
  \begin{equation}\label{eq:reduction}
  \otimes\circ (\id_{TM}\times \eta_M)= \sigma'_{M,M},
  \end{equation}
where $\sigma'$ is the co-strength of $T$, holds provided the following two diagrams commute.
\[\begin{tikzcd}
TM\times M \arrow{rr}{\id_{TM}\times\eta_M} \arrow{drr}[swap]{\eta_{TM\times M}} & & TM\times TM \arrow{d}{\sigma_{TM,M}} & & 
TM\times M \arrow{rr}{\eta_{TM\times M}} \arrow{ddrr}[swap]{\sigma'_{M,M}} & & T(TM\times M) \arrow{d}{T\sigma'_{M,M}} \\
{} & & T(TM\times M) & & & & T^2(M\times M) \arrow{d}{\mu_{M\times M}} \\
{} & & & & & & T(M\times M)
\end{tikzcd}\]
The triangle on the left commutes by the leftmost diagram in \eqref{eq:strength}. To show that the other triangle commutes, since $(\mu_M\times\id_{M})\circ (\eta_{TM}\times \id_M)=\id_{TM\times M}$, it suffices to show that the following diagram is commutative.
\[\begin{tikzcd}
{} & TM\times M \arrow{dl}[swap]{\eta_{TM}\times\id_{M}} \arrow{dr}{\eta_{TM\times M}} & {} \\
T^2M \arrow{dd}[swap]{\mu_M\times \id_M} \arrow{rr}{\sigma'_{TM,M}} \times M & & T(TM\times M) \arrow{d}{T\sigma'_{M,M}} \\
{} & & T^2(M\times M) \arrow{d}{\mu_{M\times M}} \\
TM\times M \arrow{rr}{\sigma'_{M,M}} & & T(M\times M)
\end{tikzcd}\]
In turn, the top triangle commutes by the leftmost diagram in \eqref{eq:costrength}, while the lower square commutes by the rightmost diagram in \eqref{eq:costrength}. Therefore, by equation \eqref{eq:reduction}, the commutativity of diagram \eqref{eq-item-c} is equivalent to the commutativity of the outer square below,
\[\begin{tikzcd}[column sep=2.0em, row sep=2.0em]
  TM \times M \arrow{rr}{\sigma'_{M,M}} \arrow{dd}[swap]{\id_{TM}\times r} & & T(M\times M) \arrow{dd}{T(\cdot)} \arrow{dl}[swap]{T(\id_M\times r)}\\
  {} & T(M\times M^M) \arrow{dr}{T(ev)} & \\
  TM\times M^M \arrow{ur}{\sigma'_{M,M^M}} \arrow{rr}{\epsilon} & & TM
\end{tikzcd}\]
  where $ev\colon M\times M^M\to M$ sends $(m,f)\in M\times M^M$ to $f(m)$. The upper leftmost triangle commutes by naturality of $\sigma'$, while the rightmost triangle and the lower one are easily seen to be commutative. Hence, item c) above is satisfied and diagram \eqref{eq:5} commutes, as was to be proved. Reasoning in a similar manner for the left action, one can see that $({\wT}X,TM,\intro*{\wh},\w{\rho},\w{\lambda})$ is indeed a {\swim}.
  
  It is now straightforward computations, using the \kl{commutativity} of the monad $T$, to check that the assignment
  $(X,M)\mapsto ({\wT}X,TM)$ yields the functor part of a monad on the
  category of \swims.  
\end{proof}
\begin{remark}
  \label{rem:commutativity-of-monad}
  Assume that the monad $T$ is not \kl{commutative} and we attempt to
  use in the proof of Theorem \ref{thm:lift} the monoid multiplication
  on $TM$ given by $\otimes$. All is fine for the right action and
  indeed the right action $\w{r}$ of $TM$ on itself is the unique
  extension of $\Ac{T}{M}\circ r$. However, this is not the case for
  the left action. Symmetrically, if we chose the multiplication of
  $TM$ stemming from $\otimes'$, then the left action $\w{l}$ would be
  the extension of the map $\Ac{T}{M}\circ l$, but this property would
  fail for the right action.
\end{remark}

\section{Extending the \kl{free semimodule monad} to
  \swims}\label{sec:measures-s}
In Theorem \ref{thm:lift} we showed how to lift any finitary
\kl{commutative monad} on $\Set$ to a monad on {\swimcat}.  The
purpose of the present section is then twofold. On the one hand we
provide an example of a family of $\Set$-monads to which this result
applies, and on the other hand we give explicit descriptions of the
various objects, maps and actions of the associated monads on
{\swimcat}. This will be essential for our further work on
recognisers.

Given a \kl{semiring} $S$, recall from Example~\ref{ex:semiring-monad}
the \kl{free $S$-semimodule monad} $\Sm$ on $\Set$.  Notice that $\Sm$
is a \kl{commutative monad} if and only if $S$ is a commutative
semiring, i.e., the multiplication $\cdot$ is a commutative operation.
Indeed, for a monoid $M$, the two monoid operations one can define on
$\Sm M$ are given as follows. If $f,f'\in\Sm M$, then one can define
$ff'(x)$ either by
\[
\sum_{mm'=x}{f(m)\cdot f'(m')} \textrm{\quad or\quad }
\sum_{m'm=x}{f'(m')\cdot f(m)},
\]
and the two coincide precisely when the \kl{semiring} is commutative.
\AP For this reason, for the rest of the paper we will only consider
commutative \kl{semirings} $S$.
We also consider the associated $\Set$-monad $\Sm$, along with the
profinite monad $\intro*\wS$ on $\Stone$ (cf.\
Section~\ref{sec:profinite-monads}).

Throughout the section we fix an arbitrary finite and commutative \kl{semiring} $S$. Let $X$ be a \kl{Boolean space}, and denote by $B$ its dual
algebra. Next, we provide a concrete description of the \kl{Boolean
  space} $\wS X$ in terms of \emph{\kl{measures}}
on $X$. For more details and for the proofs of several facts mentioned in this section, the interested reader is referred to \cite{R2018}.

\begin{definition}
  \AP Let $X$ be a \kl{Boolean space} and $B$ the dual algebra. An
  $S$-\emph{\kl{valued measure}} (or just a \intro{measure} when the
  \kl{semiring} is clear from the context) on $X$ is a function
  $\mu\colon B\to S$ which is finitely additive, that is
  \begin{enumerate}
  \item $\mu(0)=0$, and
  \item $\mu(K \vee L) = \mu(K)+\mu(L)$ whenever $K,L\in B$ are
    disjoint.
  \end{enumerate}
  We remark that in item $1$ the first $0$ is the bottom of the Boolean algebra,
  while the second $0$ is in $S$. Also, one can express item $2$ without
  reference to disjointness:
  \begin{enumerate}
  \item[2'.] $\mu(K\vee L)+\mu(K\wedge L)=\mu(K)+\mu(L)$ for all
    $K,L\in B$.
  \end{enumerate}
\end{definition}
Note that our notion of \kl{measure} is not standard, as we only require
finite additivity. Also, the \kl{measure} is only defined on the
clopens of the space $X$. Finally, it takes values in a (finite and 
commutative) \kl{semiring}.

\begin{notation}
  \AP Let $X$ be a set and $f\colon X\to S$ a function. If $Y\subseteq
  X$ is a subset such that the sum $\sum_{x\in Y}{f(x)}$ exists in
  $S$, then we write
  \[
  \intro{\inte_Y} f:=\sum_{x\in Y}{f(x)}.
  \]
  If $B\subseteq \P X$, and $\inte_Y f$ exists for each $Y\in
  B$, then $\inte f\colon B\to S$ denotes the function taking $Y$
  to $\inte_Y f$.
\end{notation}
  \AP Suppose $X$ is a \kl{Boolean space} and $B$ is its dual
  algebra. The Boolean algebra $\intro{\wB}$ dual to $\wS X$ is the
  subalgebra of $\P(\Sm X)$ generated by the elements of the form
  \[
  \intro{\SB{L}{k}}:=\{f\in \Sm X\mid \inte_L f=k\},
  \]
  for $L\in B$ and $k\in S$. For a proof of this fact, see \cite[Lemma~4.2]{R2018}.
Regarding the elements of $\wS X$ as Boolean algebra homomorphisms $\varphi\colon \wB\to
\2$, we can define a function
\begin{align}\label{eq:bijection-elements-measures}
  \wS X \longrightarrow \{\mu\colon B\to S\mid \mu\text{ is a \kl{measure} on } X\}, \ \ \varphi\mapsto \mu_{\varphi}
\end{align}
where $\mu_{\varphi}$ is the measure sending $L\in B$ to the unique $k\in S$ such that
$\varphi\SB{L}{k}=1$. 
In turn, the set of all \kl{measures} on $X$ is equipped with a natural topology,
generated by the sets of the form
\begin{equation}\label{eq:subbasic-measures}
\SBp{L}{k}=\{\mu\colon B\to S\mid \mu\text{ is a \kl{measure} on } X, \ \mu(L)=k\}
\end{equation}
for $L\in B$ and $k\in S$ (the notation is justified by Proposition~\ref{prop:SX-in-wSX} below). With respect to this topology, the space $\wS X$ admits the following measure-theoretic characterisation.
\begin{theorem}\label{th:profinite-measures}
  Let $S$ be a finite and commutative \kl{semiring}. For any \kl{Boolean
    space} $X$, the map in \eqref{eq:bijection-elements-measures} yields a homeomorphism between $\wS X$ and the space of $S$-valued
  \kl{measures} on $X$.
\end{theorem}
\begin{proof}
See \cite[Theorem~4.3]{R2018}.
\end{proof}

The previous result allows for a concrete representation 
of the map $\t_X$ in \eqref{eq:4} which, in turn, yields the following concrete
instantiation of Lemma \ref{l:t-dense-injective} (cf.\ also Remark
\ref{rm:injective-components-of-t}).

\begin{proposition}\label{prop:SX-in-wSX}
  If $X$ is a \kl{Boolean space}, then the function 
  \[
  \t_X\colon \Sm X\to \wS X, \ f\mapsto \inte f
  \]
  embeds $\Sm X$ in $\wS X$ as a dense subspace.  Moreover 
  $\SBp{L}{k}$, as defined in~\eqref{eq:subbasic-measures}, is the topological closure
  of $\SB{L}{k}$ whenever $L$ is a clopen of $X$, and $k\in S$.\qed
\end{proposition}
\begin{remark}
Theorem \ref{th:profinite-measures} yields, in particular, a characterisation of the
\kl{Vietoris space} $\V X$, for $X$ a \kl{Boolean space}, in terms of two-valued measures on $X$.
In \cite[Section~5]{R2018} it is also shown that, provided the semiring is idempotent (hence, a semilattice), measures can be replaced by their densities, i.e., functions $X\to S$ which are continuous with respect to the down-set topology on $S$. Thus, we recover the classical representation of the \kl{Vietoris space} of $X$
in terms of continuous functions from $X$ into the Sierpi\'{n}ski space.
\end{remark}
As follows by the general results in
Sections~\ref{sec:profinite-monads} and \ref{s:extending-set-bims},
respectively, $\wS X$ is a module over the \kl{semiring} $S$ and it is
a \kl{Boolean space with an internal monoid} if $X$ is. Here we identify the
concrete nature of this structure relative to the incarnation of $\wS
X$ as the space of \kl{measures} on $X$. We state these as lemmas and,
indeed, one can prove them directly. However, the results in this
section are just special cases of the more general results in
Sections~\ref{sec:profinite-monads} and \ref{s:extending-set-bims}.

\begin{lemma}\label{lem:addition}
  Let $X$ be a \kl{Boolean space} and let $\mu,\nu\in\wS X$. Then
  \[
  \mu+\nu\colon K\mapsto \mu(K)+\nu(K)
  \]
  is again a \kl{measure} on $X$ and the ensuing binary operation on
  $\wS X$ is continuous. Further, for any $k\in S$,
  \[
  k\mu\colon K\mapsto k\cdot\mu(K)
  \]
  is again a \kl{measure} on $X$ and the ensuing unary operation on
  $\wS X$ is continuous.\qed
\end{lemma}

This accounts for the $S$-semimodule structure on $\wS X$. Now assume
that $X$ is not just a \kl{Boolean space}, but a \swim.
To improve readability, we assume $h\colon M\to X$ is injective and
identify $M$ with its image.
Firstly, we observe that $\Sm M$ sits as a dense subspace of $\wS X$
by composing the map $\Sm h\colon\Sm M\to\Sm X$ with the integration
map of Proposition~\ref{prop:SX-in-wSX}. This is the concrete
incarnation, in the case of the monad $\Sm$, of
Lemma~\ref{l:t-dense-injective}.
\begin{lemma}\label{lem:SM-in-wSX}
  Let $(X,M)$ be a \kl{Boolean space with an internal monoid}.  Then
  \[
  \Sm M \to \wS X, \ f\mapsto \inte f
  \]
  is the map $\wh$ from~\eqref{eq:map-h-hat} transporting $\Sm M$ into
  a dense subspace of $\wS X$. \qed
\end{lemma}
We remark that, since we assumed $h$ is injective, so is the map
$\wh$ in the previous lemma (cf.\ Remark
\ref{rm:injective-components-of-t}).  Now, to exhibit the {\swim}
structure of $\wS X$, we start by identifying the actions of $M$ on
$\wS X$.

\begin{lemma}\label{lem:action of M}
  Let $(X,M)$ be a \kl{Boolean space with an internal monoid}.
  Further, let $\mu\in\wS X$ and $m\in M$. Then
  \[
  m\mu\colon K\mapsto \mu(m^{-1}K),
  \]
  where $m^{-1}K=\{x\in X\mid mx\in K\}$ whenever $K\subseteq X$ is
  clopen, is again a \kl{measure} on $X$.  This defines a left action
  of $M$ on $\wS X$ with continuous components. Similarly,
  \[
  \mu m\colon K\mapsto \mu(Km^{-1})
  \]
  defines a right action of $M$ on $\wS X$ with continuous components,
  and these actions are compatible in the sense that $(m\mu)n=m(\mu
  n)$.\qed
\end{lemma}

Using the $S$-semimodule structure of $\wS X$ (see Lemma
\ref{lem:codensity-T-alg}), along with the \kl{biaction} of $M$ on
$\wS X$ provided by the previous lemma, it is easy to obtain the
\kl{biaction} of $\Sm M$ on $\wS X$. The following can be regarded as
the specific incarnation of Theorem \ref{thm:lift}.

\begin{proposition}\label{p:SM-acts-wSX}
  Let $(X,M)$ be a \kl{Boolean space with an internal monoid}.  The
  map
  \[
  \Sm M\times \wS X \to \wS X, \ (f,\mu)\mapsto f\mu:=\sum_{m\in M}
  f(m)\cdot m\mu
  \]
  is a left action of $\Sm M$ on $\wS X$ with continuous components. A
  right action with continuous components may be defined
  similarly. These two actions are compatible and provide the
  {\swim} structure on $(\wS X, \Sm M)$.\qed
\end{proposition}

Finally, we consider a restriction of the above action of $\Sm M$ on
$\wS X$ which we will need for the construction of the space $\DssX$ in
Section \ref{s:recog-transducers}. This is given by precomposing with
the unit of the monad $\wS$:
\[
\w{\eta}_X\colon X\to \wS X, \ x\mapsto \mu_x
\]
where $\mu_x(K)=1$ if $x\in K$, and $\mu_x(K)=0$ otherwise. That is, $\mu_x= \inte \chi_x$ where $\chi_x$ is the characteristic
function of $\{x\}$ into $S$. It is immediate that this map embeds $X$ as a
(closed) subspace of $\wS X$. Thus we obtain an ``action''
\[
\Sm M\times X\to \wS X, \ (f,x)\mapsto f\mu_x.
\]
Next we observe that this ``action'' factors through the map $ \Sm
X\to \wS X$ defined in Proposition~\ref{prop:SX-in-wSX}.
\begin{lemma}\label{lem:SM-acts-X}
  Let $(X,M)$ be a \kl{Boolean space with an internal monoid}.
  Consider the map
  \[
  \Sm M\times X \to \Sm X, \ (f,x)\mapsto fx,
  \]
  where $fx(y):=\sum_{mx=y}f(m)$. Then we have
  \[
  f\mu_x= \inte fx.
  \]
  Furthermore, for each $f\in \Sm M$, the assignment $x\mapsto \inte
  fx$ is continuous.\qed
\end{lemma}

\section{Recognisers for operations given by $S$-valued
  transductions}\label{s:recog-transducers}
In this section we will see how we can use the extension of a
$\Set$-monad $T$ to {\swims} obtained in
Section~\ref{s:extending-set-bims} to generate recognisers for
languages obtained by applying an operation modelled by the monad $T$.

It is by now a standard result in the theory of formal languages that
many operations on languages can be modelled using transductions,
i.e., maps of the form $M\to\P N$ for two monoids $M$ and $N$,
see~\cite{PiSa82}. The starting point of this work is the observation
that the existential quantifier can also be modelled as a transduction,
as we will see in Section~\ref{sec:S-transduction}. Furthermore,
\kl{modular quantifiers} $\modexists$ of modulus $q$ fit into the same
pattern. The only difference is that instead of using transductions of
the form $M\to\P N$ one needs to replace the powerset $\P N$ with
the free $\Zq$-semimodule over $N$. More generally, we are interested
in operations that can be modelled as maps $M\to\Sm N$ with $\Sm$
denoting as before the free $\Sm$-semimodule monad. In category theory
these maps are known as \kl{Kleisli maps} for $\Sm$, the morphisms in the so-called Kleisli
category of $\Sm$.

We start Section~\ref{sec:T-transduction} by briefly recalling the
definition of the Kleisli maps for a monad. Then we present the
blueprint of our approach, using an additional assumption on the
$T$-Kleisli map under consideration (namely that it is a monoid
morphism), and in Section~\ref{sec:S-transduction} we instantiate $T$
to the \kl{free $S$-semimodule monads} for commutative \kl{semirings}
$S$ and we adapt the general theory developed previously.

\subsection{Recognising operations modelled by a monad $T$}
\label{sec:T-transduction}

\AP Consider a monad $(T,\eta,\mu)$ on a category $\C$. The
\intro{Kleisli category} $\intro\Kl(T)$ for $T$ is equivalent to the
category of \emph{free} \kl{$T$-algebras} and has played a crucial
r\^ole in program semantics for modelling functions with
side-effects. Formally, the objects of $\Kl(T)$ are the objects
in the underlying category $\C$ and a morphism $X\to Y$ in $\Kl(T)$
(called a $T$-\intro{Kleisli map}) is a morphism $X\to TY$ in
$\C$. One can think of an object $X$ in $\Kl(T)$ as the generator of
the free algebra $TX$. Notice that morphisms $X\to Y$ in $\Kl(T)$ are
in one-to-one correspondence with \kl{$T$-algebra} morphisms
$TX\to TY$ between the corresponding free algebras.

Hereafter we assume $T$ is an arbitrary \kl{commutative} and
\kl{finitary} monad on $\Set$, and let $A,B$ be finite sets. We start
by observing that a \kl{Kleisli map} $R\colon \As\to T(\Bs)$ could be
used to transform languages in the alphabet $B$ into languages in the
alphabet $A$. Assume that $L=\phi^{-1}(P)$ for some monoid morphism
$\phi\colon \Bs\to M$ and some $P\subseteq M$.  We consider the
function
\[
\begin{tikzcd}
  \As\ar[r,"R"] & T(\Bs)\ar[r,"T\phi"] & TM.
\end{tikzcd}
\]
Since $T$ is a \kl{commutative monad}, we know that it lifts to the
category of monoids and thus we can see $T\phi$ as a monoid
morphism. If $R$ is also a monoid morphism, and we will assume this
only in this subsection, then so is $T\phi\circ R$, and it could be
used for language recognition in the standard way. Assuming that we
have a way of turning the recognising sets in $M$ into recognising
sets in $TM$, i.e., that we have a predicate transformer $\P M\to \P
TM$ mapping $P$ to $\widetilde{P}$, we obtain a language
$\widetilde{L}$ in $\As$ as the preimage of $\widetilde{P}$ under the
morphism $T\phi\circ R$.

\begin{remark}
  \label{rem:1}
  In the running example of the next subsection we will need maps $R$
  that are not monoid morphisms, and in that setting we will have to
  use a matrix representation of the transduction
  instead. Nevertheless, the techniques used in the next subsection can
  be seen as an adaptation of the theory developed here for the case
  when $R$ is indeed a monoid morphism.
\end{remark}

In this work we go beyond regular languages, so we are interested in
languages \kl{recognised} by a {\swim} \kl{morphism} as follows:
\begin{equation}
  \label{eq:9}
  \begin{tikzcd}
    \kl{\beta}(\Bs) \ar[r,"\wphi"] & X \\
    \Bs \ar[r,"\phi"]\ar[u] & M\ar[u,swap,"h"]
  \end{tikzcd}
\end{equation}

We recall that to improve readability, and since ${\wphi}$ is uniquely
determined by its restriction to $\Bs$, we sometimes denote such a
\kl{morphism} of \swims\ simply by $\phi$, instead of $(\wphi,\phi)$.

By Theorem~\ref{thm:lift}, we know that $({\wT}X,TM)$ is a {\swim},
and in what follows we use it for \kl{recognising} $A$-languages by
constructing another {\swim} \kl{morphism}
$(\kl{\beta}(\As),\As)\to({\wT}X,TM)$ as in
Lemma~\ref{lem:Kleisli-BiM} below. To this end, we need a way of
lifting the \kl{Kleisli map} $R\colon \As\to T(\Bs)$ to a \kl{Kleisli
  map} for the monad $\wT$. This can be done in a natural way using a
natural transformation
\begin{align*}
  \thash\colon\kl{\beta} T\to\wT\kl{\beta}
\end{align*}
obtained from the natural transformation $\t_X\colon T|X|\to |\wT X|$
defined in~\eqref{eq:4} using the unit $\iota$ and the counit $\epsilon$
of the adjunction $\beta\dashv\fgt$. Explicitly, $\intro*\thash$ is obtained
as the composite
\begin{equation}
  \label{eq:7}
\begin{tikzcd}
  \beta T\ar[r,"\beta T \iota"] & \beta T\fgt \beta\ar[r,"\kl{\beta}\t\kl{\beta}"] &\beta\fgt\wT\beta\ar[r,"\epsilon\wT\beta"]& \wT\beta\,.
\end{tikzcd}
\end{equation}
(This is a rather standard construction in category theory, see for
example~\cite[Theorem~9]{Street}). In down-to-earth terms, the component of $\thash$ at the set $Y$ is the free extension of the composite
\[\begin{tikzcd}
TY \arrow{r}{T\iota_Y} & {T|\beta Y|} \arrow{r}{\t_{\beta Y}} & {|\wT\beta Y|}.
\end{tikzcd}\] 
It follows that the natural
transformation $\thash\colon \kl{\beta} T\to {\wT}\kl{\beta}$ also
behaves well with respect to the units and multiplications of the
monads. That is, in the terminology of~\cite{Street}, the pair
$(\kl{\beta},\thash)$ is a \emph{monad opfunctor}.  This in turn
implies that $\kl{\beta}$ can be lifted to a functor $\intro{\wbeta}$
between the Kleisli categories, making the next square commute, where
the vertical functors are the free functors from the base to the
Kleisli categories.
\begin{equation*}
  \begin{tikzcd}
    {\Kl(T)}\ar[r,"\wbeta"] & \Kl({\wT})  \\
    \Set\ar[u]\ar[r,"\kl{\beta}"] & \Stone\ar[u]
  \end{tikzcd}
\end{equation*}
The functor $\wbeta$ maps the \kl{Kleisli map} $R\colon \As\to T(\Bs)$
to the \kl{Kleisli map}
$\w{R}\colon \kl{\beta}(\As)\to \wT\kl{\beta}(\Bs)$ given by the
composite
  \begin{equation}\label{eq:R-hat}
  \begin{tikzcd}
    \w{R}\colon\kl{\beta}(\As)\ar[r,"\kl{\beta} R"] & \kl{\beta}
    T(\Bs)\ar[r,"\thash"] &{\wT} \kl{\beta} (\Bs)\,.
  \end{tikzcd}
  \end{equation}

\begin{lemma}
  \label{lem:Kleisli-BiM}
  Assume $R\colon \As\to T(\Bs)$ is a monoid morphism. If the pair $({\wphi},\phi)$ from~\eqref{eq:9} is a \kl{morphism} of
  \swims, then so is the pair $({\wT}{\wphi}\circ\w{R}, T\phi\circ R)$
  described in the next diagram.
  \[
  \begin{tikzcd}
    {\kl{\beta}(\As)}\ar[r,"{\w{R}}"] &{{\wT}\kl{\beta} (\Bs)} \ar[r,"{{\wT}{\wphi}}"] & {{\wT}X} \\
    \As\ar[u,"\iota"]\ar[r,"R"]& T(\Bs) \ar[r,"T\phi"] &
    TM\ar[u,swap,"\t_X\circ Th"]
  \end{tikzcd}
  \]

\end{lemma}
\begin{proof}
  In the statement of the lemma we have omitted writing the forgetful
  functor $\fgt$ on the top line of the diagram. We will need it
  nevertheless in the proof. Using the definition of $\w{R}$, we need
  to show that the next diagram commutes:
  \[
  \begin{tikzcd}
    {|\kl{\beta}(\As)|}\ar[r,"{|\kl{\beta} R|}"] & {|\kl{\beta} T(\Bs)|}\ar[r,"|\thash|"] &{|{\wT}\kl{\beta} (\Bs)|} \ar[r,"{|{\wT}{\wphi}|}"] & {|{\wT}X|} \\
    & & T|\kl{\beta} (\Bs)|\ar[r,"{T|{\wphi}|}"]\ar[u,"\t\beta"] & T |X|\ar[u,"\t"]\\
    \As\ar[uu,"\iota"]\ar[r,"R"]&
    T(\Bs)\ar[uu,"\iota T"]\ar[ur,"T\iota"] \ar[r,"T\phi"] &
    TM\ar[ur,swap,"Th"] &
  \end{tikzcd}
  \]
  The two rectangles in the diagram above commute by naturality of
  $\iota$, respectively $\t$, and the bottom right rhombus commutes
  because $\phi$ is a \kl{morphism} of \swims. To prove that the
  middle trapezoid is commutative, we just have to recall how the
  transformation $\thash$ is defined, see~\eqref{eq:7}. In a
  2-categorical terminology, this is a simple exercise involving the
  mates $\t$ and $\thash$:
\[
\begin{tikzcd}[column sep=35pt]
  \fgt\beta T \ar[r,"\fgt\beta T\iota "]& \fgt\beta T\fgt\beta \ar[r,"\fgt\beta\t\beta"]& \fgt\beta\fgt\wT\beta \ar[r,"\fgt\epsilon\wT\beta"]&\fgt\wT\beta \\
  T \ar[r,swap,"T\iota "]\ar[u,"\iota T"]& T\fgt\beta \ar[r,swap,"\t\beta "]\ar[u,"\iota T\fgt\beta"]& \fgt\wT\beta \ar[u,"\iota\fgt\wT\beta"]\ar[ru,swap,"\id"] & 
\end{tikzcd}
\]
The squares commute by the naturality of $\iota$, whilst the triangle
commutes because $\fgt\epsilon\circ\iota\fgt=\id$.
\end{proof}

\subsection{Recognising quantified languages via $S$-transductions}
\label{sec:S-transduction}
Here we show how to construct {\swims} \kl{recognising} quantified
languages.  We point out that the content of this subsection could be
easily adapted to arbitrary \kl{Kleisli maps} for the monads of the
form $\wS$, for commutative \kl{semirings} $S$.
We start with a language $L$ in the extended alphabet $\Ats$
\kl{recognised} by a {\swim} \kl{morphism} as in the following diagram.
\begin{equation*}
  \begin{tikzcd}
    \kl{\beta} (\Ats) \ar[r,"{\wphi}"] & X \\
    \Ats \ar[r,"\phi"]\ar[u] & M\ar[u,swap,"h"]
  \end{tikzcd}
\end{equation*}
In other words, there exists a clopen $C$ in $X$ such that
$L=\phi^{-1}(C\cap M)$. The aim of this subsection is to construct
\kl{recognisers} for the quantified languages $\Lex$ and $\Lmodex$, as
defined in Section~\ref{sec:logic-words}.  To this end, using the
formal sum notation in the definition of the monad $\Sm$, we consider
the map $R\colon \As\to \Sm(\Ats)$ defined by
\[
w\mapsto \sum_{i=1}^{\len{w}}{1_S \cdot \mw{i}}.
\]
If $S$ is the Boolean semiring $2$, then $R$ simply associates with each
word $w$ the set of all words in $\Ats$ with the same shape as $w$ and
with exactly one \kl{marked letter}. The framework developed in the
previous subsection does not immediately apply, since $R$ is not a
monoid morphism. So the first step we have to take is to obtain a
monoid morphism from $R$, which will then be used to construct {\swim}
\kl{recognisers} for quantified languages.

Upon viewing $R$ as an $S$-transduction (see~\cite{Sakarovitch09}), we
observe that it is realised by the rational $S$-transducer $\TR$
pictured in Figure~\ref{fig:transducer}, in which we have drawn
transition maps only for a generic letter $a\in A$.

\begin{figure}[h]
  \centering
  \begin{tikzpicture}[->,>=stealth',shorten >=1pt,auto,node
    distance=2.4cm, semithick, scale=0.4]
    \tikzstyle{every state}=[draw]q

    \node[initial,state] (A) {$1$}; \node[state] (B) [right of=A]
    {$2$}; \node[draw=none] (C) [right of=B,xshift=-15mm,] {};

    \path (A) edge [loop above] node {$a|a$} (A) edge node {$a|a'$}
    (B) (B) edge [loop above] node {$a|a$} (B) (B) edge (C);
  \end{tikzpicture}
  \caption{The $S$-transducer $\TR$ realising $R$. All the transitions
    have weights $1_S$, and thus the transducer outputs value $1_S$
    for all pairs of the form $(w,\mw{i})$, with $w\in \As$ and $1\le
    i\le\len{w}$. }
  \label{fig:transducer}
\end{figure}
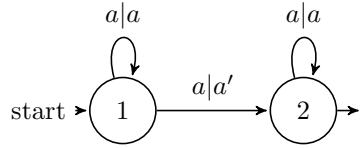

\AP This transducer provides the following representation of $R$ in
terms of a monoid morphism
\begin{equation}
  \label{eq:15}
  \Rmon\colon \As\to\mat{2}{\Sm(\Ats)},  
\end{equation}
where $\intro{\mat{n}{\Sm(\Ats)}}$ denotes the set of $n\times
n$-matrices over the \kl{semiring} $\Sm(\Ats)$. For a word $w\in \As$, the
matrix $\Rmon(w)$ has at position $(i,j)$ the formal sum of output
words obtained from the transducer $\TR$ by going from state $i$ to
state $j$ while reading input word $w$. That is, $\Rmon$ is given by
\[
w\mapsto
\begin{pmatrix}
  1_S\cdot \wz & \sum_{i}{1_S\cdot \mw{i}} \\
  0_S & 1_S\cdot \wz\\
\end{pmatrix}.
\]
The next two examples provide the motivation for considering the
particular transduction $R$ in the first place.

\begin{example}
  \label{exmp:lexists}
  Assume $S$ is the Boolean semiring $\2$, thus $\Sm=\Pfin$, and
  $R(w)=\{\mw{i}\mid 1\le i\le\len{w}\}$. The language $\Lex\subseteq
  \As$ is \kl{recognised} by the following composite monoid morphism,
  that will be denoted by $\phi_\exists$.
  \[
  \begin{tikzcd}
    \As\ar[r,"\Rmon"] & \mat{2}{\Pfin((A\times
      2)^*)}\ar[rr,"\mat{2}{\Pfin\phi}"] & & \mat{2}{\Pfin M}
  \end{tikzcd}
  \]
  Indeed, if $L=\phi^{-1}(P)$ for some $P\subseteq M$ then
  $\Lex=\phi_{\exists}^{-1}(\widetilde{P})$, where $\widetilde{P}$ is
  the set of matrices in $\mat{2}{\Pfin M}$ such that the finite set
  in position $(1,2)$ intersects $P$.
\end{example}
\begin{example}
  \label{exmp:lmodexists}
  Assume $S$ is the \kl{semiring} $\Zq$. The language
  $\Lmodex\subseteq \As$ is \kl{recognised} by the following composite
  monoid morphism, that will be denoted by $\phi_\modexists$.
  \[
  \begin{tikzcd}
    \As\ar[r,"\Rmon"] & \mat{2}{\Sm((A\times
      2)^*)}\ar[rr,"\mat{2}{\Sm\phi}"] & & \mat{2}{\Sm M}
  \end{tikzcd}
  \]
  Indeed, if $L = \phi^{-1}(P)$ with $P\subseteq M$ then
  $\Lmodex=\phi_{\modexists}^{-1}(\widetilde{P})$, where
  $\widetilde{P}$ is the set of matrices in $\mat{2}{\Sm M}$ such that
  the finitely supported function $f\colon\Zq\to M$ in position
  $(1,2)$ has the property that $\inte_P f= p$ in $\Zq$.
\end{example}

In view of Theorem~\ref{thm:lift}, we know that whenever $(X,M)$ is a
\swim, then so is $(\wS X,SM)$ with the actions of the internal monoid
as in Proposition~\ref{p:SM-acts-wSX}. Using this fact, we can prove the following lemma.

\begin{lemma}
  If $(X,M)$ is a \swim, then so is
  \begin{align*}
    (\mat{n}{\wS X},\mat{n}{\Sm M})
  \end{align*}
  for any integer $n\geq 1$.
\end{lemma}
\begin{proof}
  The set $\mat{n}{\wS X}$ is a \kl{Boolean space} with
  respect to the product topology of $n\times n$ copies of $\wS X$.
  The statement then follows easily upon defining the actions of the
  monoid $\mat{n}{\Sm M}$ on $\mat{n}{\wS X}$ by using the actions of
  $\Sm M$ on $\wS X$ via matrix multiplication, and the $S$-semimodule
  structure of $\w{S}X$. For example, the left action of
  $(f_{ij})_{i,j}\in\mat{n}{\Sm M}$ on $(\mu_{ij})_{i,j}\in\mat{n}{\wS
    X}$ yields a matrix of \kl{measures} in $\wS X$ having at position
  $(i,j)$ the measure $\sum_{k=1}^n{f_{ik}\mu_{kj}}$.
\end{proof}

We next prove a result which entails that the monoid morphisms
$\phi_\exists$ and $\phi_\modexists$ constructed in
Examples~\ref{exmp:lexists} and~\ref{exmp:lmodexists} can be extended
to {\swim} \kl{morphisms} \kl{recognising} $\Lex$ and $\Lmodex$,
respectively.

\begin{lemma}
  \AP If the pair $({\wphi},\phi)$ from~\eqref{eq:9} is a
  \kl{morphism} of {\swims} and $\Rmon\colon \As\to\mat{n}{\Sm(\Bs)}$
  is a monoid morphism, then the pair
  $(\mat{n}{{\wS}{\wphi}}\circ\wRmon, \mat{n}{\Sm\phi}\circ \Rmon)$
  described in the next diagram is a {\swim} morphism,
  \[
  \begin{tikzcd}[column sep=4em]
    \kl{\beta}(\As)\ar[r,"\wRmon"] &\mat{n}{\wS \kl{\beta}(\Bs)} \ar[r,"\mat{n}{\wS {\wphi}}"] & \mat{n}{\wS X}\\
    \As\ar[u]\ar[r,"\Rmon"]& \mat{n}{\Sm(\Bs)} \ar[r,"\mat{n}{\Sm
      \phi}"] & \mat{n}{\Sm M}\ar[u,swap,"\mat{n}{\t_X\circ \Sm h}"]
  \end{tikzcd}
  \]
  where $\intro{\wRmon}$ is the unique continuous extension of the following composite map:
  \[\begin{tikzcd}
    \As \arrow[rr,"\Rmon"] & & \mat{n}{\Sm(\Bs)}
    \arrow[rr,"\mat{n}{\iota}"] & & \mat{n}{\kl{\beta}\Sm(\Bs)}
    \arrow[rr,"\mat{n}{\thash}"] & & \mat{n}{\wS \kl{\beta}(\Bs)}.
  \end{tikzcd}\]
\end{lemma}
\begin{proof}
This follows essentially by Lemma \ref{lem:Kleisli-BiM} by setting $T=\Sm$, along with the functoriality of $\mat{n}{-}$. Note that the aforementioned lemma applies to this setting because $\Rmon$ is a monoid morphism.
\end{proof}
If we apply the previous lemma to the monoid morphism $\Rmon$ in equation
\eqref{eq:15} we obtain the {\swim} $(\mat{2}{\wS X},\mat{2}{\Sm M})$
which, when instantiated with the appropriate \kl{semiring} $S$,
\kl{recognises} the quantified languages $\Lex$ and $\Lmodex$.

For instance, suppose the \kl{semiring} $S$ is $\Zq$.  If $L$ is
\kl{recognised} by a clopen $C\subseteq X$ then, upon recalling from \eqref{eq:subbasic-measures} that
subbasic clopens of $\wS X$ are of the form $\SBp{K}{k}$ for $K$ a clopen of
$X$ and $k\in S$, one can easily prove that the quantified language
$\Lmodex$ is \kl{recognised} by the clopen subset of $\mat{2}{\wS X}$
given by the product $ \wS X \times \SBp{C}{p}\times \wS X\times \wS
X$, where the elements of the clopen $\SBp{C}{p}$ should appear in
position $(1,2)$ in the matrix view of the space.

However, notice that the image of the morphism $\mat{2}{\Sm
  {\wphi}}\circ \wRmon$  
  is contained in the subspace of $\mat{2}{\wS X}$ which can be represented
by the matrix
\[
\begin{pmatrix}
  X &\wS  X \\
  0 & X \\
\end{pmatrix}.
\]
As a consequence, we can use for the same recognition purpose a
smaller \swim, through which the morphism $\mat{2}{\Sm{\wphi}}\circ \wRmon$ 
factors.
\AP We denote this {\swim} \kl{morphism} by
\[
\intro*{\Dssphi}\colon (\kl{\beta}(\As),\As)\to(\DssX,\DssM),
\]
where 
\[
\DssX:= \wS X\times X \ \text{and} \ \intro*{\DssM}:=\Sm M\times M,
\]
with monoid structure and \kl{biactions} defined essentially by
identifying the products above with upper triangular matrices, and
then using the matrix multiplication and the concrete description of
several monoid actions from Lemmas~\ref{lem:action of M}
and~\ref{lem:SM-acts-X}. Using the notations described in these
lemmas, the left action of $\DssM$ on $\DssX$ can be described by
\[
\begin{pmatrix}
  m & f \\
  0 & m \\
\end{pmatrix}
\begin{pmatrix}
  x & \mu \\
  0 & x \\
\end{pmatrix}
=
\begin{pmatrix}
  mx & m\mu+\inte fx \\
  0 & mx \\
\end{pmatrix},
\]
where $(f,m)\in \DssM$ and $(\mu,x)\in \DssX$.
Recall from Section \ref{sec:logic-words} that the language $\Qk(L)$ in the
alphabet $A$ is obtained by quantifying the language $L\subseteq \Ats$
with respect to the quantifier associated with a \kl{semiring} $S$ and an element $k\in S$. We summarise the preceding observations in
the following theorem.
\begin{theorem}\label{th:recognition-quantified-languages}
  Let $S$ be a commutative \kl{semiring}, and $k\in S$. Suppose a language
  $L\subseteq \Ats$ is recognised by the {\swim} \kl{morphism}
  $\phi\colon (\kl{\beta}(\Ats),\Ats)\to(X,M)$. Then the quantified
  language $\Qk(L)\subseteq \As$ is recognised by the {\swim}
  \kl{morphism} $\Dssphi\colon (\kl{\beta}(\As),\As)\to(\DssX,\DssM)$.\qed
\end{theorem}
As an immediate consequence, taking $S=2$ the Boolean semiring and
$k=1$, we recover the result in \cite[Proposition~13]{GehrkePR16} on
existential quantification:
\begin{corollary}
  Consider a formula $\varphi(x)$ with a free first-order variable $x$. If
  the language $L_{\varphi(x)}\subseteq \Ats$ is recognised by the
  {\swim} \kl{morphism} $\phi\colon (\kl{\beta}(\Ats),\Ats)\to(X,M)$,
  then the existentially quantified language $L_{\exists
    x.\varphi(x)}\subseteq \As$ is recognised by the {\swim}
  \kl{morphism} $\Ds_2 \phi\colon (\kl{\beta}(\As),\As)\to(\V X\times
  X,\Pfin M\times M)$.\qed
\end{corollary}

\section{Duality-theoretic account of the
  construction}\label{sec:duality-expl}

Let $S$ be a finite and commutative \kl{semiring}, and $(X,M)$ a
{\swim}. As earlier, we denote by $B$ the dual algebra of
$X$. Further, let $\phi\colon (\kl{\beta}(\Ats), \Ats) \to (X,M)$ be a
{\swim} \kl{morphism}.  We denote by $\intro*\B$ the preimage under
$\phi$ of $B$. That is, $\B$ is the Boolean algebra, closed under
quotients in $\P(\Ats)$, of languages \kl{recognised} by the {\swim}
\kl{morphism} $\phi$.

\AP In Section \ref{sec:S-transduction} we introduced the map $\Dssphi$
as a \kl{recogniser} for the quantified languages obtained from the
languages in $\B$.  
Here we prove, by duality, that $\Dssphi$ is in fact the minimal possible {\swim} recogniser for these quantified languages. This will allow us to get a Reutenauer-type theorem for $\DssX$, see Theorem \ref{t:reutenauer}. The idea is the following. On the language side, we are interested in the Boolean algebra generated by the languages of the form $\Qk(L)$, for $k\in S$ and $L\in\B$. This coincides with the Boolean algebra $\QB$ obtained as the preimage of $\wB$, the Boolean algebra of clopens of $\wS X$, under the composite
\[\begin{tikzcd}
\phiQ\colon {A^*} \arrow{r} & \Sm(\Ats) \arrow{r} & \Sm M \arrow{r} & {\wS X}, 
\end{tikzcd}\]
where the first map sends $w\in A^*$ to $\sum_{i=1}^{\len{w}}{1_S \cdot \mw{i}}$, the second one is $\Sm\phi\colon \Sm(\Ats) \to\Sm M$, and the third one is the integration map (cf.\ Lemma \ref{lem:SM-in-wSX}).
Indeed, suppose $\phi^{-1}(K)=L\in \B$ for $K$ a clopen of $X$. Then, for every $k\in S$,
\begin{align*}
\phiQ^{-1}(\SBp{K}{k})&=\{w\in A^*\mid \int_K{\sum_{i=1}^{\len{w}}{1_S \cdot \phi(\mw{i})}}=k\} \\
&=\{w\in \As\mid w\in \Qk(\phi^{-1}(K))\}= \Qk(L). 
\end{align*}
The Boolean algebra $\QB$ is not closed under quotients. Since we want a {\swim} recogniser, and not just a ``Boolean space
recogniser'', we want to recognise the Boolean algebra closed under quotients generated by $\QB$.
Furthermore, from the viewpoint of logic we are adding one layer of quantifiers. Thus, by inductive hypothesis, it makes sense to include also the languages of the form $L_0=\{w\in \As\mid \wz\in L\}$, for $L\in \B$. 
These are the languages in the alphabet $A$ which are recognised by $\phi$ upon composing with the embedding
\[
\embz\colon \As\to \kl{\beta}(\Ats), \ w\mapsto \wz.
\]
Let $\Bzero$ be the Boolean algebra that is the preimage of $\B$ under $\embz$. 
We thus want a {\swim} recogniser for $\Bprime$, the closure under quotients of ${<}\QB\cup\Bzero{>}_{\BA}$. We show that:
\begin{enumerate}
\item The Boolean algebra ${<}\QB\cup\Bzero{>}_{\BA}$ is already closed under quotients, whence $\Bprime={<}\QB\cup\Bzero{>}_{\BA}$.
\item This allows us to see $\Bprime$ as a quotient of the coproduct of $\QB$ and $\Bzero$, hence also of $\wB$ and $B$. By describing the quotienting operations on these subalgebras we can define a compatible quotienting operation on the coproduct, which makes the natural map  $\wB+B\to \P(\As)$ a homomorphism of Boolean algebras with biactions. 
\item Finally, dualising the quotienting operation on $\wB+B$ we get the actions of $\DssM$ on $\DssX$ given by matrix multiplication in Section~\ref{sec:S-transduction}. Further, we recover $\Dssphi$ as dual to the homomorphism $\wB+B\to \P(\As)$.
\end{enumerate}

To improve readability, throughout this section we omit reference to the semiring $S$, and write $\intro{\Dphi}, \intro{\DsX}, \intro{\DsM}$ instead of $\Dssphi, \DssX, \DssM$.

\subsection{The space $\DsX$ by duality}\label{s:diamond-X-by-duality}
\AP Recall from \eqref{eq:R-hat} the Kleisli map $\w{R}$, and consider the continuous map
\[
\intro*{\phiQ}\colon\kl{\beta}(\As)\xrightarrow{\w{R}}
\wS\kl{\beta}(\Ats) \xrightarrow{\wS\phi} \wS X
\]
which is given for $w\in \As$ by
\[
\intro*{\phiQ}(w):=\inte \fw,
\]
where
\[
  \intro*{\fw}:=\sum_{1\leq i\leq\len{w}}1_S \cdot \phi(\mw{i}).
\]
For any $k\in S$ and
$L\in\B$, the clopen in $\kl{\beta}(\As)$ corresponding to
$\Qk(L)$ is $\phiQ^{-1}(\SBp{K}{k})$, where $K\subseteq X$ is the
clopen in $X$ \kl{recognising} $L$ via $\phi$ and $\SBp{K}{k}$ is as in equation \eqref{eq:subbasic-measures}.  By
Theorem~\ref{th:profinite-measures}, the clopens of $\wS X$ are
generated by the sets of the form $\SBp{K}{k}$ with $k\in S$ and
$K\subseteq X$ clopen, thus we have:
\begin{proposition}
  \AP The Boolean algebra $\intro*\QB$ of those languages over $A$
  which are inverse images of clopens under $\phiQ$ is generated by the quantified languages
  $\Qk(L)$, for $k\in S$ and $L\in\B$.\qed
\end{proposition}
Note that $\QB$, as defined in the previous proposition, is \emph{not} closed under quotients. This is the
reason we had to make an adjustment between
Sections~\ref{sec:T-transduction} and~\ref{sec:S-transduction} above.

\AP We denote by $\intro*\Bzero$ the Boolean algebra of languages
closed under quotients which is \kl{recognised} by the {\swim}
\kl{morphism}
\[
\intro*{\phiz}\colon(\kl{\beta}(\As),\As)\xrightarrow{\embz}(\kl{\beta}(\Ats),\Ats)\xrightarrow{\phi} (X,M).
\]
\AP In other words, $\Bzero$ consists of the languages of the form $\intro{\Lz}:=\phiz^{-1}(K)$,
obtained as the preimage under $\embz$ of languages $L=\phi^{-1}(K)$
in $\B$. Taking the product map, it now follows that
\[
\Dphi=\phiQ\times\phiz\colon\kl{\beta} (\As)\to \wS X\times
X,
\]
viewed just as a map of \kl{Boolean spaces}, ``recognises'' the
Boolean algebra generated by $\QB\cup\Bzero$, in the sense that the elements of the latter Boolean algebra are exactly those of the form $\Dphi^{-1}(C)$ for some clopen $C\subseteq \wS X\times X$. However, since $\QB$ is
\emph{not} closed under quotients, a priori, neither is
${<}\QB\cup\Bzero{>}_{\BA}$.

The Boolean algebra $\Bprime$ that we are interested in is the closure
under quotients of ${<}\QB\cup\Bzero{>}_{\BA}$. The important
observation is that ${<}\QB\cup\Bzero{>}_{\BA}$ \emph{is already
  closed under the quotient operations}, thus explaining why $\wS
X\times X$, along with the above product map, is the right recogniser
space-wise.

\begin{proposition}\label{prop:B'}
  The Boolean algebra generated by $\QB\cup\Bzero$ is closed under
  quotients. That is,
  \[
  \Bprime={<} \Qk(L), \Lz\mid L\in\B \text{ and }k\in S {>}_{\BA}.
  \]
\end{proposition}
\begin{proof}
  Since $\Bzero$ is closed under quotients, it suffices to consider
  the quotienting of languages of the form
  $\Qk(L)=\phiQ^{-1}(\SBp{K}{k})$ where $K\subseteq X$ is the clopen
  \kl{recognising} $L$ via $\phi$. For $u\in \As$ we have
  \begin{align*}
    u^{-1} \Qk(L)
    & =  \{w\in \As\mid  uw\in \Qk(L)\} \\
    & = \{w\in \As\mid \inte f_{uw}\in\SBp{K}{k}\}.
  \end{align*}
 Since the free variable in the word $uw$ either
  occurs in $u$ or in $w$,
  \[
  f_{uw}=\phi(\uz)\fw+\fu\phi(\wz).
  \]
  Further, since $\inte(\phi(\uz)\fw+\fu\phi(\wz)) =
  \inte\phi(\uz)\fw+ \inte \fu\phi(\wz)$, we have
  \begin{align*}
    u^{-1} \Qk(L)&= \{w\in \As\mid   \inte  \phi(\uz)\fw+ \inte \fu\phi(\wz)\in\SBp{K}{k} \} \\
    &=\bigcup_{k_1+k_2=k} \{w\in\As \mid \inte
    \phi(\uz)\fw\in\SBp{K}{k_1} \ \text{and} \ \inte
    \fu\phi(\wz)\in\SBp{K}{k_2}\}.
  \end{align*}
  Now,
  \begin{equation}\label{eq:11}
    \inte  \phi(\uz)\fw\in\SBp{K}{k_1}\ \iff\ \inte \fw\in\SBp{\phi(\uz)^{-1}K}{k_1}
  \end{equation}
  which in turn is equivalent to $w\in \Qkp((\uz)^{-1}L)$, which is an
  element of $\QB$. We now proceed with the second condition. We have
  $ \inte \fu\phi(\wz) \in\SBp{K}{k_2}$ if, and only if, there is a
  set
  \[
  I\subseteq{\rm Sup}(\fu):=\{m\in M\mid \fu(m)\neq 0\}
  \]
  with
  \begin{itemize}
  \item $ \inte _I\fu=k_2$;
  \item $m\phi(\wz)\in K$ for each $m\in I$;
  \item $m\phi(\wz)\not\in K$ for each $m\not\in I$.
  \end{itemize}
  Observe that $m\phi(\wz)\in K$ if, and only if, $w\in
  \phiz^{-1}(m^{-1}K)$. Thus \[\{w\in \As\mid \inte
  \fu\phi(\wz) \in\SBp{K}{k_2}\}\] is equal to
  \begin{align}\label{eq:12}
    \bigcup_{\substack{I\subseteq{\rm Sup}(\fu) \\
        \inte _I\fu=k_2}} \bigg( [\bigcap_{m\in I}
      \phiz^{-1}(m^{-1}K)] \cap [\bigcap_{m\in I^c}
      \phiz^{-1}(m^{-1}K^c)] \bigg)
  \end{align}
  which is in $\Bzero$.
\end{proof}
\begin{corollary}
  The dual space of $\Bprime$ is a closed subspace of $\wS X \times
  X$. In particular, $\Bprime$ is \kl{recognised} as a Boolean algebra
  by $\wS X \times X$.
\end{corollary}

\begin{proof}
  By the previous proposition,
  $\Bprime={<}\QB\cup\Bzero{>}_{\BA}$. But $\Bzero$ is exactly the
  preimage of the dual of $X$ under $\phiz$, and $\QB$ is exactly the
  preimage of the dual of $\wS X$ under $\phiQ$. Thus $\Bprime$ is
  exactly the preimage of the dual of $\wS X\times X$ under $\Dphi$,
  and therefore $\Bprime$ is \kl{recognised} as a Boolean algebra by
  $\wS X\times X$.

  Now, factoring the map $\Dphi$, we obtain a closed subspace $Y$ of
  $\wS X\times X$ so that
  \[
  \Dphi\colon\kl{\beta} (\As)\twoheadrightarrow Y\hookrightarrow\wS
  X\times X.
  \]
  Since the dual of the quotient map $\kl{\beta}
  (\As)\twoheadrightarrow Y$ is an embedding whose image is $\Bprime$,
  the dual of $Y$ is isomorphic to $\Bprime$.
\end{proof}

\subsection{The internal monoid structure of $\DsX$ by duality}

In Section \ref{sec:S-transduction} the monoid operation of $\DsM=\Sm M\times M$ and the actions of $\DsM$ on $\DsX$ were given in terms of matrix multiplication. This multiplication was introduced in an ad hoc manner. Here we show that these actions (and, in particular, the monoid operation) need not be guessed, as they can be derived by duality. In fact, they are the appropriate actions on $\DsX$ for making $\Dphi$ a {\swim} morphism.
For this purpose, we consider the homomorphism dual to $\Dphi$:
\[
\varphi\colon\wB+ B\to\P(\As), \
\SBp{K}{k}\mapsto\phiQ^{-1}(\SBp{K}{k}), \ K\mapsto\phiz^{-1}(K).
\]
We already know, by Proposition~\ref{prop:B'}, that the image of
$\varphi$ is closed under quotients. The point is, in fact, that
Proposition~\ref{prop:B'} tells us that we can define a biaction of $\DsM$ on
$\wB+ B$ so that $\varphi$ becomes a homomorphism of Boolean algebras
with biactions. Thus, for each $(f,m)\in\DsM$, we want to define a
``left quotient'' by $(f,m)$ (that is, the component at $(f,m)$ of a
right action) on $\wB+ B$ (and a ``right quotient'', which is a left
action) so that $\varphi$ becomes a homomorphism of Boolean algebras
with biactions.

The monoid morphism from $\As$ to $\DsM$ is given by sending the
internal monoid element $u\in \As$ to the internal monoid element
$(\fu,\phi(\uz))\in\Sm M\times M$, where $\fu$ is defined as at
the beginning of Section \ref{s:diamond-X-by-duality}. 
Now, the component at $(f,m)$ of a ``left
quotient'' operation on $\wB+ B$ is a homomorphism
\[
\Lambda(f,m)\colon\wB+ B\to \wB+ B.
\]
Given the nature of coproducts, such a homomorphism is determined by
its components $\Lambda_1(f,m)\colon\wB\to \wB + B$ and
$\Lambda_2(f,m)\colon B\to \wB + B$. Our goal then, is to show that:
\begin{itemize}
\item the computation of quotient operations in the image of $\varphi$
  combined with wanting $\varphi$ to be a morphism of Boolean algebras
  with biactions, dictates what $\Lambda_1(f,m)$ and $\Lambda_2(f,m)$
  must be;
  \item the left action of $\DsM$ on $\DsX$ dual to $\Lambda$ coincides with the one defined in Section \ref{sec:S-transduction}.
\end{itemize}
The symmetric facts for the right action are similar and thus we only
consider the left action. Also, note that we will not prove directly
that the $\Lambda(f,m)$'s that we define are components of a right
action on a Boolean algebra, as this will follow from the second bullet point above. 

So, we want to define the action such that $\varphi$ becomes a
homomorphism sending the action of $(\fu,\phi(\uz))$ to the action of
the quotient operation $u^{-1}(\ )$ on $\P(\As)$. The computations in
the proof of Proposition~\ref{prop:B'} tell us the components of
$u^{-1}\phiQ^{-1}(\SBp{K}{k})$ in $\QB$ and in $\Bzero$,
respectively. Since $\QB$ and $\Bzero$ are precisely the images under
$\varphi$ of $\wB$ and $B$, respectively, the computation tells us how
to define $\Lambda_1(\fu,\phi(\uz))$ using components
$\Lambda_{11}(f,m) \colon\wB\to\wB$ and $\Lambda_{12}(f,m)\colon\wB\to
B$.

By the computation in~\eqref{eq:11}, we have that the component
$\Lambda_{11}(f,m)\colon\wB\to\wB$ depends only on the second
coordinate of the pair $(\fu,\phi(\uz))$ and it sends $\SBp{K}{k}$ to
$\SBp{(\phi(\uz))^{-1}K}{k}$. Stating it for an arbitrary element
$(f,m)\in\Sm M\times M$, we have
\[
\Lambda_{11}(f,m)\colon\wB\to\wB, \ \SBp{K}{k}\mapsto
\SBp{m^{-1}K}{k}.
\]
Similarly, the computation in~\eqref{eq:12}, stated for an arbitrary
element $(f,m)\in\Sm M\times M$, yields $\Lambda_{12}(f,m)\colon\wB\to
B$ given by
\begin{align}\label{eq:12gen}
  \SBp{K}{k}&\mapsto \bigcup_{\substack{I\subseteq{\rm Sup}(f)\\ \inte
      _If=k}}\bigg([\bigcap_{n\in I} n^{-1}K] \cap [\bigcap_{n\in I^c}
    n^{-1}K^c]\bigg).
\end{align}

The above observations imply that
\begin{proposition}\label{prop:leftquotient}
  The map $\varphi\colon\wB+ B\to\P(\As)$ is a homomorphism of Boolean
  algebras with biactions when the left quotient operation
  $\Lambda(f,m)$ of $\wB+ B$ is defined on $\wB$ by
  \[
  \Lambda_1(f,m)\colon \SBp{K}{k} \mapsto
  \bigvee_{k_1+k_2=k}(\Lambda_{11}(\SBp{K}{k_1})\wedge\Lambda_{12}(\SBp{K}{k_2}))
  \]
  and on $B$ by $\Lambda_2(f,m)\colon K \mapsto m^{-1}K$.\qed
\end{proposition}
It is now an easy verification that the maps $\Lambda_{11}(f,m)$ and
$\Lambda_{12}(f,m)$ are dual to the summands of the first component of
the action of $(f,m)$ on $\DsX$, and that $\Lambda_1(f,m)$ and
$\Lambda_2(f,m)$ are dual, respectively, to
\[
\lambda_1(f,m)\colon\wS X\times X\to\wS X,\quad (\mu,x)\ \mapsto\
m\mu+{\int}fx,
\]
and
\[
\lambda_2(f,m)\colon\wS X\times X\to X,\quad (\mu,x)\ \mapsto\ mx.
\]
\begin{lemma}\label{lem:11}
  The homomorphism $\Lambda_{11}(f,m)\colon\wB\to\wB$ given by
  $\SBp{K}{k}\mapsto \SBp{m^{-1}K}{k}$ is dual to the continuous
  function $\lambda_{11}(f,m)\colon\wS X\to \wS X$ given by
  $\mu\mapsto m\mu$, where
  \[
  m\mu\colon B\to S, \ K\mapsto \mu(m^{-1}K).
  \]
\end{lemma}

\begin{proof}
  The function $\lambda_{11}(f,m)$ is dual to $\Lambda_{11}(f,m)$ if,
  and only if, for all $\mu\in\wS X$ and all $\SBp{K}{k}\in\wB$ we
  have
  \[
  \lambda_{11}(f,m)\mu\in\SBp{K}{k}\quad\iff\quad\mu\in\Lambda_{11}(f,m)
  \SBp{K}{k}.
  \]
  But $\lambda_{11}(f,m)\mu=m\mu$, so
  \begin{align*}
    \lambda_{11}(f,m)\mu\in\SBp{K}{k}\quad&\iff\quad m\mu\in\SBp{K}{k}\\
    &\iff\quad m\mu(K)=k\\
    &\iff\quad \mu(m^{-1}K)=k\\
    &\iff\quad \mu\in\SBp{m^{-1}K}{k}= \Lambda_{11}(f,m),
  \end{align*}
  as was to be proved.
\end{proof}

\begin{lemma}\label{lem:12}
  The homomorphism $\Lambda_{12}(f,m)\colon\wB\to B$ given as
  in~\eqref{eq:12gen} is dual to the continuous function
  $\lambda_{12}(f,m)\colon X\to \wS X$ given by $x\mapsto \inte fx$,
  where
  \[ \inte fx \colon B\to S, \ K\mapsto
  \inte_{\hskip-5ptKx^{-1}}\hskip-7pt f.
  \]
\end{lemma}
\begin{proof}
  Let $x\in X$ and $[K,k]\in\wB$. Then
  \begin{align*}
    \inte fx\in[K,k] \quad &\iff \quad \inte_K fx=k\\
    &\iff\quad \inte_{Kx} f=k\\
    &\iff\quad \sum_{x\in n^{-1}K} f(n)=k,
  \end{align*}
  and the latter is true if, and only if, there exists $I\subseteq
  {\rm Sup}(f)$ with $\inte _If=k$ and $x\in n^{-1}K$ for each $n\in
  I$ and $x\not\in n^{-1}K$ for each $n\not\in I$. That is,
  \[
  \inte fx\in[K,k] \quad \iff \quad x\in \Lambda_{12}(f,m)[K,k].
  \]
Therefore, the homomorphism $\Lambda_{12}(f,m)$ is dual to the continuous map $\lambda_{12}(f,m)$.
\end{proof}
\begin{lemma}\label{lem:1}
  The homomorphism $\Lambda_{1}(f,m)\colon\wB\to \wB+B$ given as in
  Proposition~\ref{prop:leftquotient} is dual to the continuous
  function $\lambda_{1}(f,m)\colon \wS X\times X\to \wS X$ given by
  $(\mu,x)\mapsto m\mu+\inte fx$.
\end{lemma}
\begin{proof}
  Let $(\mu,x)\in\wS X\times X$ and $[K,k]\in\wB$. Then
  \begin{align*}
    \lambda_1(f,m)(\mu,x)\in[K,k] \quad &\iff \
    \lambda_{11}(f,m)\mu\,+\,
    \lambda_{12}(f,m)x\in[K,k]\\
    &\iff\ \exists k_1,k_2\left(k_1+k_2=k,
      \lambda_{11}(f,m)\mu\in[K,k_1],
      \text{ and }\lambda_{12}(f,m)x\in[K,k_2]\right)\\
    &\iff\ \exists k_1,k_2\left(k_1+k_2=k,
      \mu\in\Lambda_{11}(f,m)[K,k_1],
      \text{ and }x\in\Lambda_{12}(f,m)[K,k_2]\right)\\
    &\iff\ (\mu,x)\in\Lambda_1[K,k],
  \end{align*}
  as was to be shown.
\end{proof}
\noindent It is straightforward that $\Lambda_2(f,m)\colon K \mapsto
m^{-1}K$ is dual to $\lambda_2(f,m)\colon x\mapsto mx$. 
Since the continuous map $\lambda_1(f,m)\times \lambda_2(f,m)\colon\DsX\to\DsX$ coincides with
the component at $(f,m)$ of the left action of $\DsM$ on $\DsX$ defined in Section \ref{sec:S-transduction} through matrix multiplication,
we conclude that
\begin{corollary}\label{cor:action-duality}
  The dual of the left quotienting operation $\Lambda$ on $\wB+ B$ defined in
  Proposition~\ref{prop:leftquotient} is the left action of
  $\DsM$ on $\DsX$ defined in Section \ref{sec:S-transduction}.\qed
\end{corollary}

A similar result holds of course for the right action, and the monoid operation of $\DsM$ can be recovered by restricting the actions on $\DsX$ to $\DsM$. As a
consequence, we have

\begin{theorem}\label{thm:duality}
  Let $\phi\colon(\kl{\beta}(\Ats),\Ats)\to (X,M)$ be a {\swim}
  \kl{morphism}.  Then the
  homomorphism of Boolean algebras with biactions
  \[
  \varphi\colon\wB+ B\to\P(\As), \ \ 
  \SBp{K}{k}\mapsto\phiQ^{-1}(\SBp{K}{k}), \ \ K\mapsto\phiz^{-1}(K),
  \]
  obtained by equipping $\wB+ B$ with the \kl{biaction} of $\DsM$ as
  indicated in Proposition~\ref{prop:leftquotient}, is dual to the {\swim} \kl{morphism}
  \[
  \Dphi\colon (\kl{\beta}(\As),\As)\to(\DsX,\DsM)
  \]
  derived in Section~\ref{sec:S-transduction}.\qed
\end{theorem}

\subsection{A Reutenauer theorem for $\DsX$}
In this last subsection we prove a Reutenauer-type result (see Theorem \ref{t:reutenauer} below)
characterising the Boolean algebra closed under quotients generated by
all languages \kl{recognised} by the space $\DsX$ through
\emph{\kl{length preserving}} \kl{morphisms}.
\begin{definition}
  \AP We call a {\swim} \kl{morphism}
  $\psi\colon(\kl{\beta}(\As),\As)\to(\DsX,\DsM)$ \intro{length
    preserving} provided, for each $a\in A$, we have that
  \[
  \pi_1\circ\psi(a)\colon M\to S
  \]
  is the characteristic function $\chi_{m_a}$ for some single $m_a\in
  M$. That is, $\pi_1\circ\psi(a)(m_a)=1$ and $\pi_1\circ\psi(a)(m)=0$
  for all $m\in M$ with $m\neq m_a$.
\end{definition}

Recall that, given any {\swim} \kl{morphism} $\phi\colon
(\kl{\beta}(\Ats),\Ats)\to(X,M)$, we obtain a {\swim} \kl{morphism}
\[
\Dphi\colon(\kl{\beta}(\As),\As)\to(\DsX,\DsM),\quad w\mapsto \bigg( \inte
\fw, \phi(\wz)\bigg).
\]
Upon defining $f_a:=\pi_1\circ\Dphi(a)$, we have $f_a=\chi_{m_a}$ where $m_a=\phi(a,1)$.
Hence, $\Dphi$ is \kl{length preserving}. It is now a matter of a
straightforward computation to prove the following proposition.
\begin{proposition}\label{prop:lengthpres}
  Let $X$ be a {\swim}. Every \kl{length preserving} {\swim}
  \kl{morphism} $(\kl{\beta}(\As),\As)\to(\DsX,\DsM)$ is of the form
  $\Dphi$ for some {\swim}
  \kl{morphism} $\phi \colon (\kl{\beta}(\Ats),\Ats) \to (X,M)$.
\end{proposition}
\begin{proof}
  Consider an arbitrary \kl{length preserving} {\swim} \kl{morphism}
  $\psi\colon(\kl{\beta}(\As),\As)\to(\DsX,\DsM)$. We define $\phi
  \colon (\kl{\beta}(\Ats),\Ats) \to (X,M)$ by
  \begin{align*}
    \phi\colon \Ats &\to M, \\
    (a,0) &\mapsto \pi_2\circ\psi(a)\\
    (a,1) &\mapsto m_a
  \end{align*}
  where $m_a\in M$ is such that $\pi_1\circ\psi(a)=\chi_{m_a}$. The
  universal property of the Stone-{\v C}ech compactification guarantees
  that $\phi$ is a {\swim} \kl{morphism} with the topological
  component $\wphi=\kl{\beta}\phi$. It now suffices to show that
  $\psi(a)=\Dphi(a)$ for each $a\in A$:
  \begin{align*}
    \Dphi(a)&=(f_a,\phiz(a))=(\chi_{\phi(a,1)},\phi(a,0))=(\chi_{m_a},\pi_2\circ\psi(a))=(\pi_1\circ\psi(a),\pi_2\circ\psi(a))=\psi(a).\qedhere
  \end{align*}
\end{proof}

\begin{theorem}\label{t:reutenauer}
  Let $X$ be a {\swim}, and $A$ a finite alphabet. The Boolean
  subalgebra closed under quotients of $\P(\As)$ generated by all
  languages over $A$ which are \kl{recognised} by a \kl{length
    preserving} {\swim} \kl{morphism} into $\DsX$ is generated as a
  Boolean algebra by the languages over $A$ \kl{recognised} by $X$,
  and the languages $\Qk(L)$ for $L$ a language over $A\times 2$
  \kl{recognised} by $X$.
\end{theorem}
\begin{proof}
  Let us denote by $\mathcal{B}''$ the Boolean algebra generated by
  the languages over $A$ \kl{recognised} by $X$, and the languages
  $\Qk(L)$ for $L$ a language over $A\times 2$ \kl{recognised} by $X$.

  If $L'\in\P(\As)$ is \kl{recognised} by a \kl{length preserving}
  {\swim} \kl{morphism}
  $\psi\colon(\kl{\beta}(\As),\As)\to(\DsX,\DsM)$, then by
  Proposition~\ref{prop:lengthpres} there is a {\swim} \kl{morphism}
  $\phi \colon (\kl{\beta}(\Ats),\Ats) \to (X,M)$ such that
  $\Dphi=\psi$. That is, $L'$ lies in the Boolean algebra called
  $\Bprime$ in the beginning of this section. Since $\Bprime\subseteq
  \mathcal{B}''$ by Proposition~\ref{prop:B'}, we have $L'\in
  \mathcal{B}''$.

  For the other inclusion, if $L$ is a language over $A\times 2$
  \kl{recognised} by $X$, then $\Qk(L)$ is \kl{recognised} by $\DsX$
  through a length preserving morphism in view of
  Theorem~\ref{th:recognition-quantified-languages}. Finally, suppose
  $L$ is a language over $A$ \kl{recognised} by
  $\eta\colon\kl{\beta}(\As)\to X$ through the clopen $K$.  Consider
  any function $\phi\colon\kl{\beta}(\Ats)\to M$ satisfying
  $\phi(a,0)=\eta(a)$ for each $a\in A$. Then $L=\Dphi^{-1}(\wS
  X\times K)$, showing that $L$ is \kl{recognised} by $\DsX$ through a
  length preserving morphism.
\end{proof}

\section{Conclusion}

In this paper we provide a general construction for recognisers which
captures the action of quantifier-like operations on languages of
words, drawing heavily on a combination of general categorical
tools and more concrete duality theoretic ones. 

This paper is a stepping stone in a long-term research programme
aimed at finding meaningful ultrafilter equations that characterise
logically defined classes of non-regular languages. The next step is
to understand the effect on equations of the constructions introduced
here.

The generic development of Section~\ref{s:recog-transducers} allows
this work to be extended to encompass a wider range of operations on
languages modelled by rational transducers which, by the
Kleene-Sch{\"u}tzenberger theorem (see, e.g., \cite{Sakarovitch09}),
admit a matrix representation. Also, the duality-theoretic account in
Section~\ref{sec:duality-expl} leads to a Reutenauer-type
characterisation theorem, akin to the one in~\cite{GehrkePR16}.  It
would also be interesting to understand a common framework for our
contributions and the recent work~\cite{ChenU16}.

 \section*{Acknowledgement} 
 We are grateful to Jean-{\'E}ric Pin for suggesting the transducer approach, to Thomas Colcombet for sharing his \texttt{knowledge} package, and to the DuaLL team for helpful discussions on duality and logic on words.

\bibliographystyle{plain}

\end{document}